\documentclass[11pt]{article}

\usepackage[T1]{fontenc}
\usepackage[utf8]{inputenc}
\usepackage{lmodern}

\usepackage{amsmath}
\usepackage{amssymb}
\usepackage{amsthm}
\usepackage{mathtools}

\usepackage[margin=1in]{geometry}
\usepackage{enumitem}
\usepackage{booktabs}
\usepackage{microtype}
\usepackage{needspace}
\usepackage{float}

\usepackage[dvipsnames]{xcolor}

\definecolor{ParetoBlue}{HTML}{444444}
\definecolor{ParetoOrange}{HTML}{B96518}
\definecolor{ParetoGray}{HTML}{66717C}

\usepackage{tikz}
\usetikzlibrary{
  arrows.meta,
  calc,
  patterns,
  positioning
}

\usepackage{pgfplots}
\pgfplotsset{compat=1.18}

\usepackage{aliascnt}

\usepackage[
  colorlinks = true,
  linkcolor  = NavyBlue,
  citecolor  = ForestGreen,
  urlcolor   = NavyBlue
]{hyperref}

\usepackage[nameinlink, capitalize]{cleveref}

\newtheorem{theorem}{Theorem}[section]

\newaliascnt{lemma}{theorem}
\newtheorem{lemma}[lemma]{Lemma}
\aliascntresetthe{lemma}

\newaliascnt{proposition}{theorem}
\newtheorem{proposition}[proposition]{Proposition}
\aliascntresetthe{proposition}

\newaliascnt{corollary}{theorem}
\newtheorem{corollary}[corollary]{Corollary}
\aliascntresetthe{corollary}

\theoremstyle{definition}

\newaliascnt{definition}{theorem}

\aliascntresetthe{definition}

\newaliascnt{construction}{theorem}
\newtheorem{construction}[construction]{Construction}
\aliascntresetthe{construction}

\theoremstyle{remark}

\newaliascnt{remark}{theorem}

\aliascntresetthe{remark}

\newaliascnt{conjecture}{theorem}

\aliascntresetthe{conjecture}

\crefname{construction}{construction}{constructions}
\Crefname{construction}{Construction}{Constructions}

\newcommand{\E}{\mathbb E}
\newcommand{\Prb}{\mathbb P}
\newcommand{\vol}{\operatorname{vol}}
\newcommand{\ind}{\mathbf 1}
\newcommand{\Max}{\operatorname{Max}}

\newcommand{\proofstep}[1]{%

  \smallskip%
  \noindent%
  \textit{#1}\ %
}

\title{The Pareto Secretary Problem\thanks{%
A preliminary version of this paper has been accepted to the
22nd Conference on Web and Internet Economics (WINE 2026).}}

\author{
  Andr\'es Cristi\textsuperscript{1}
  \and
  Yilin Kong\textsuperscript{2}
  \and
  Jiechen Zhang\textsuperscript{1}
}

\date{}

\hypersetup{
  pdftitle  = {The Pareto Secretary Problem},
  pdfauthor = {Andrés Cristi, Yilin Kong, Jiechen Zhang}
}

\begin{document}

\maketitle

\footnotetext[1]{%
  EPFL.
  \texttt{\{andres.cristi, jiechen.zhang\}@epfl.ch}%
}

\footnotetext[2]{%
  The Chinese University of Hong Kong.
  \texttt{1155191882@link.cuhk.edu.hk}%
}

\setcounter{footnote}{2}

\begin{abstract}
The Pareto secretary problem asks for online selection of an observation that is maximal in several criteria. Observations arrive as an i.i.d. sample, and success means that no other observation is at least as good in every coordinate and strictly better in one. For every fixed dimension $d\ge2$ and independent continuous coordinates, we prove that the optimal failure probability is asymptotic to $\Theta_d(1/[(\log n)^{d-1}\log\log n])$, both with full information and with coordinate ranks alone. A deterministic rank rule attains this bound uniformly over the marginals. Under arbitrary dependence, a fixed joint density with uniform marginals can instead have optimal success tending to the scalar best-choice constant $\approx 0.58$. For every fixed positive definite Gaussian covariance $\Sigma$, the failure probability is also $\Theta_d(1/[(\log n)^{d-1}\log\log n])$, attained by a parameter-free ordinal algorithm. While the leading constant depends on $\Sigma$, this shows that the correlation between the dimensions has no first-order effect. For both independent coordinates and nonsingular Gaussian distributions, we determine the optimal fixed cutoff and its leading failure constant. The best cutoff rule has a failure probability larger than the unrestricted optimum by a factor $\Theta_d((\log\log n)^2)$. For singular Gaussian distributions with known correlations, we also characterize the optimal failure asymptotics achievable using coordinate ranks alone. We measure each point by the probability that a new observation dominates it, and use the scarcity of points with very small probabilities to bound the performance of every stopping rule. A matching rule selects a point not dominated by earlier observations if this probability is below a threshold that increases as time runs out. An initial sample lets us estimate these domination probabilities using coordinate ranks alone, accurately enough to attain the optimal leading failure probability under both independence and Gaussian dependence. Singular Gaussian dependence can make every observation acceptable, or bring back the difficulty of selecting a single best point.
\end{abstract}

\clearpage
\setcounter{tocdepth}{2}

{
  \small
  \tableofcontents
}

\clearpage

\section{Introduction}

The classical secretary problem asks a decision maker to select the best candidate from a sequence arriving in uniformly random order. Each candidate must be accepted or rejected immediately, and a rejected candidate cannot be recalled. When the decision maker observes only relative ranks, the optimal probability of selecting the best candidate tends to $1/e$~\cite{ferguson1989}. Observing numerical values drawn i.i.d.\ from a known continuous distribution improves the limit to approximately $0.58$ in the full-information best-choice problem~\cite{gilbert1966,gnedin2023}. In both settings, selecting the single best candidate leaves a constant probability of failure even as the number of candidates grows.

Several criteria can change which outcomes are acceptable. A hiring committee may compare candidates by education and professional experience without assigning a common numerical score to the two attributes. One candidate may have more education, while another has more experience. These candidates represent different tradeoffs, and either can be a reasonable choice when the criteria are assessed together. This setting leads naturally to a selection objective that allows several mutually incomparable outcomes.

We represent a candidate by a point in $\mathbb R^d$, with larger coordinates preferable. One point dominates another if it is at least as good in every coordinate and strictly better in at least one. A point is \emph{Pareto maximal} if no other point in the sample dominates it; the set of all such points is the \emph{Pareto frontier}. We study a known number $n$ of independent observations $Z_1,\ldots,Z_n$ from the same distribution, revealed sequentially, and seek to select one point on the frontier of the complete sample. Selection is irrevocable, and at most one point can be selected. Selecting a Pareto-maximal point is a success; selecting any other point or selecting nothing is a failure. When $d=1$ and the distribution is continuous, the frontier almost surely consists of the unique largest observation. In several dimensions it can contain many points, each offering a different tradeoff among the criteria.

This additional freedom changes the asymptotic outcome. For every fixed $d\ge2$, independent continuous coordinates permit success probability tending to one. Gnedin~\cite{gnedin1983pareto} already showed that this probability can tend to one in a model where observations are ranked independently and uniformly at random in two coordinates. The growing set of acceptable points creates opportunities for reliable online selection. Our results quantify how its size governs the optimal failure probability and identify how dependence between the criteria changes that conclusion.

\subsection{Main results}

We compare two information settings. A \emph{full-information} rule knows the joint distribution and observes numerical values. A \emph{rank-based} rule observes only the ranking of all observations seen so far in each coordinate. Below, we specify what distributional knowledge the rank-based rules need. We consider fixed dimensions $d\ge2$ and fixed distributions as $n\to\infty$.

The expected number of acceptable observations gives a common way to express our main result. Let $M_n$ be the expected number of Pareto-maximal points in the complete sample. For independent continuous coordinates and for Gaussian observations with a fixed positive definite covariance matrix, the optimal failure probability satisfies
\[
1-p_n^\star\sim\frac{2}{M_n\log M_n},
\]
where $p_n^\star$ denotes the optimal success probability in either information setting. Here $\sim$ means that the ratio of the two sides tends to one as $n\to\infty$. Thus, within these families, a larger expected frontier makes online selection easier in a precise quantitative sense. The constructions below attain this leading term using coordinate ranks alone; the Gaussian rule also learns the correlations from the observations. Section~\ref{gauss:comparison-independent} derives the common relation from the distribution-specific bounds.

\paragraph{Independent coordinates.}
We first assume that the coordinates of each observation are independent and have continuous distributions, which may differ across coordinates. A simple cutoff rule rejects the first $s$ observations and selects the first later point that no earlier point dominates. We characterize exactly which cutoffs make the probability of selecting a maximal point tend to one. The smallest failure probability among these rules is of order
\[
\frac{\log\log n}{(\log n)^{d-1}}
\]
(\Cref{thm:fixed-d-successful-cutoff-families,thm:cutoff-rate-lower-bound}).

We then improve this bound with a time-dependent threshold rule. It selects the first point that no earlier observation dominates and whose probability of being dominated by a new independent observation is below a threshold specified for its arrival time. This threshold increases as fewer observations remain, allowing more points to be accepted near the end. The failure probability improves to order
\[
\frac{1}{(\log n)^{d-1}\log\log n}.
\]
We determine the leading constants for both rules, which depend only on the dimension. A rank-based implementation of the more selective rule achieves the same leading failure probability as the optimal full-information rule, for every choice of continuous marginal distributions (\Cref{opt:thm:optimal-rate}). It needs to know only that the coordinates are independent.

\paragraph{Allowing dependence.}
We next allow dependence between coordinates, while observations remain independent. Without restrictions on this dependence, the improvement over one-dimensional selection can disappear. We construct a fixed distribution with a joint density and uniform marginals whose optimal full-information probability tends to $\approx 0.5801$, the tight success probability for the one-dimensional case~\cite{gilbert1966}. This is the sharp asymptotic worst-case constant among distributions with a joint density, since using a single coordinate always gives at least the one-dimensional full-information probability (\Cref{dep:thm:main}).

For Gaussian observations with a fixed positive definite covariance matrix $\Sigma$, both constructions extend from the independent case with the same orders of failure probability. Correlation enters through the size of the frontier. We write $a_\Sigma$ for its leading coefficient, so that
\[
M_n\sim a_\Sigma(\log n)^{d-1}.
\]
A larger coefficient means a larger expected frontier at the same sample size, to first order. It equals $1/(d-1)!$ under independence and depends only on the correlation matrix in the Gaussian model. Section~\ref{sec:gaussian} derives an integral formula for this coefficient and explains its role in both stopping rules. We determine the leading constants for the best cutoff and the optimal rule (\Cref{gauss:cutoff-consistency,gauss:headline}). A rank-based rule again matches the optimal full-information leading term, knowing only that the observations are Gaussian, without knowing their means, variances, or correlations.

We also consider singular Gaussian distributions, whose observations lie in a lower-dimensional affine subspace, and characterize how the relationships between coordinates affect the selection guarantees (\Cref{gauss:singular-gaussian}).

\subsection{Proof overview}

For independent coordinates, we apply each marginal distribution function and then the inverse standard Gaussian distribution function to obtain independent standard Gaussian coordinates.
Because we assume the distributions are absolutely continuous, this preserves independence and, almost surely, the ranking in every coordinate and hence which points are maximal. It gives us a common starting point for the later extension to Gaussian dependence.

Our rules never select a point already dominated by an earlier observation. They can therefore fail in two ways, by selecting nothing or by selecting a point that a later observation dominates. For a cutoff rule, waiting longer leaves fewer opportunities to select, but also fewer later observations that might dominate the selected point. We estimate both probabilities and balance them to find the best cutoff.

A more selective rule accepts only points that a new observation is unlikely to dominate. Such points are rare, so a rule that selects with high probability must sometimes accept a point that a later observation dominates. This gives a lower bound on failure. Our time-dependent threshold rule matches it. We bound how often an observation meets the threshold without being dominated by earlier observations, and hence how likely the rule is still waiting at each time. We then add the probabilities of selecting a point that a later observation dominates, bounding all failure near the end by the probability of still waiting.

To implement the rule using only coordinate rankings, we reject the first half of the observations and use them as a sample. For each later observation, we use the number of sample observations exceeding it in each coordinate to estimate the probability of being exceeded in that coordinate. Under independence, the product estimates its probability of being dominated by a new observation. We prove that these estimates are accurate enough to give the same leading failure probability as with known distributions.

With Gaussian dependence, a linear transformation makes the coordinates independent, but changes the directions in which one point dominates another (Figure~\ref{fig:intro-gaussian-change}). Probability estimates for the new comparison regions let us extend both the fixed-cutoff analysis and the lower and upper bounds on the optimal failure probability from the independent case. The rank-based rule also estimates the correlations from the initial sample.

 \begingroup
\setlength{\intextsep}{10pt}
\setlength{\abovecaptionskip}{8pt}
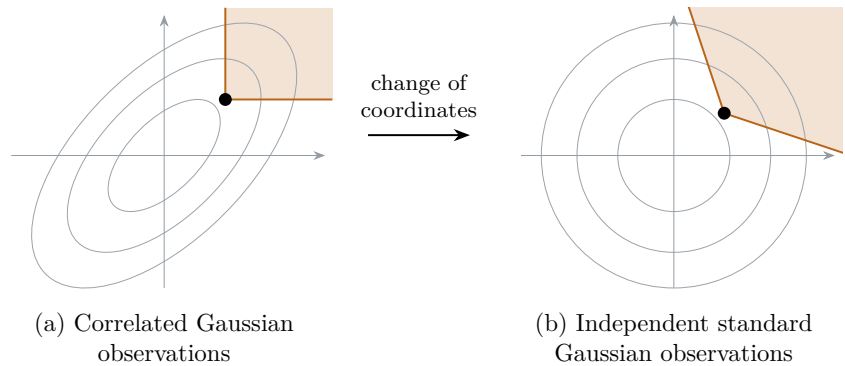
\begin{figure}[H]
\centering
\scalebox{0.93}{%
\begin{tikzpicture}[scale=1.45,font=\small,>=Stealth]
\begin{scope}
 \clip (-1.65,-1.45) rectangle (1.65,1.45);
 \fill[ParetoOrange!18] (.6,.55) rectangle (1.65,1.45);
 \begin{scope}[cm={.948683,.316228,.316228,.948683,(0,0)}]
  \foreach \r in {.55,.95,1.3}
   \draw[ParetoGray!65] (0,0) circle (\r);
 \end{scope}
 \draw[->,ParetoGray!70] (-1.5,0)--(1.58,0);
 \draw[->,ParetoGray!70] (0,-1.3)--(0,1.38);
 \draw[ParetoOrange,thick] (.6,1.45)--(.6,.55)--(1.65,.55);
 \fill (.6,.55) circle (1.8pt);
\end{scope}
\node[align=center,font=\small] at (0,-1.78) {(a) Correlated Gaussian\\observations};
\draw[->,thick] (2.0,.2)--(3.0,.2);
\node[anchor=south,align=center,text width=1.9cm,inner sep=2pt,font=\footnotesize] at (2.5,.32) {change of\\coordinates};
\begin{scope}[xshift=5cm]
 \begin{scope}
  \clip (-1.65,-1.45) rectangle (1.65,1.45);
  \fill[ParetoOrange!18] (.494106,.415049)--(6.186204,-1.482319)
    --(4.288836,4.209779)--(-1.403262,6.107147)--cycle;
  \foreach \r in {.55,.95,1.3}
   \draw[ParetoGray!65] (0,0) circle (\r);
  \draw[->,ParetoGray!70] (-1.5,0)--(1.58,0);
  \draw[->,ParetoGray!70] (0,-1.3)--(0,1.38);
  \draw[ParetoOrange,thick] (6.186204,-1.482319)--(.494106,.415049)--(-1.403262,6.107147);
  \fill (.494106,.415049) circle (1.8pt);
 \end{scope}
 \node[align=center,font=\small] at (0,-1.78) {(b) Independent standard\\Gaussian observations};
\end{scope}
\end{tikzpicture}%
}
\caption{A change of coordinates makes the Gaussian coordinates independent. The shaded regions identify the same dominating observations; after the change, comparisons follow the slanted region.}
\label{fig:intro-gaussian-change}
\end{figure}
\endgroup

The other results use separate arguments. For arbitrary dependence, we construct a distribution where the additional coordinates provide no useful alternative to selecting the first-coordinate maximum. To compare values with ranks in the worst case, we adapt the method of Correa et al.~\cite{correa2021}. We change numerical values without changing their order, making a rule's decisions depend almost entirely on ranks.

\subsection{Related work}

The early literature on stopping at a Pareto-optimal option includes Gnedin~\cite{gnedin1983pareto}; Ferguson~\cite{ferguson1992best} discusses this line of multicriteria stopping problems. Gnedin models the criteria as independent uniform random rankings and proves that, for two criteria, a fixed-cutoff rule can select a Pareto-maximal observation with probability tending to one. Independent continuous marginals induce the same rank distribution, so this conclusion also applies to our rank-based setting. The fixed-cutoff rule is the same as the one we analyze under independence. His result establishes the existence of cutoffs whose selection probability tends to one; we characterize all such cutoff sequences and determine the best cutoff failure asymptotics. Our analysis determines the sharp first-order failure asymptotics, compares numerical and rank information, and relates stopping difficulty to expected frontier size under independent coordinates and Gaussian dependence.

Another multiple-criteria objective accepts an observation that is best in at least one criterion. Gnedin~\cite{gnedin1981multicriterial} studies this objective using relative ranks under independent criteria. With two criteria, at most two observations are acceptable. Pareto maximality also accepts intermediate tradeoffs that maximize neither coordinate, and the expected frontier size grows with the sample under independence. Ferguson~\cite{ferguson1992best} and Gnedin~\cite{gnedin1994best} study dependent-criteria variants. Samuels and Chotlos~\cite{samuels1986multiple} instead minimize the expected sum of final coordinate ranks, with rank one best in each coordinate. Yu, Jacobson, and Kiyavash~\cite{yu2019bi} consider multiple selections with full information and two expected-reward objectives, studying policies that are Pareto optimal with respect to those objectives.

Secretary problems on partially ordered sets share the objective of selecting a maximal element. Freij and W\"astlund~\cite{freij2010} obtain a probability at least $1/e$ of selecting a maximal element for every finite partially ordered set in random arrival order. Kumar et al.~\cite{kumar2011} give guarantees improving with the number of maximal elements. Garrod and Morris~\cite{garrod2011} obtain a sharp improving bound when the width, the largest number of mutually incomparable elements, equals the number of maximal elements; they also obtain a $1/e$ guarantee without this restriction. These results connect selection guarantees to the number of acceptable points in a partially ordered set. Here the order is generated by an i.i.d.\ geometric sample, and the decision maker also has access to coordinate comparisons. This structure permits the sharp relation between expected frontier size and optimal failure in \eqref{eq:frontier-cutoff-optimum}. Partial comparisons also arise in selection models motivated by biased evaluations and fairness, with score-based or group-based objectives~\cite{salem2023,pmlr-v139-correa21a}.

The unknown-marginal reduction adapts the Ramsey-theoretic method of Correa et al.~\cite{correa2021} for making stopping decisions insensitive to numerical labels. Their objective maximizes the expected selected value relative to the expected sample maximum; ours maximizes the probability of selecting a Pareto-maximal point. Applying Ramsey theory to coordinate grids gives increasing changes of numerical values that preserve all coordinate comparisons. The argument allows dependence and ties, including singular distributions. The complete reduction appears in \Cref{app:rank-sufficiency}.

\paragraph{Organization.}
\Cref{sec:model} gives the model and information settings. \Cref{sec:fixed-cutoffs,sec:optimal-stopping} analyze fixed cutoffs and optimal stopping for independent coordinates. \Cref{sec:arbitrary-dependence} studies arbitrary dependence, and \Cref{sec:gaussian} extends the analysis to Gaussian dependence and singular cases, followed by discussion. The appendices contain the rank-sufficiency reduction and the probability and sample estimates used in the stopping proofs.

\section{Model and Information Settings}
\label{sec:model}

\subsection{Model}

Fix a horizon $n\in\mathbb N$ and a dimension $d\ge2$. The observations $Z_1,\ldots,Z_n$ are independent with common distribution $F$ on $\mathbb R^d$. Unless stated otherwise, we assume $F$ is absolutely continuous with respect to the Lebesgue measure. For $z,z'\in\mathbb R^d$, write $z\preceq z'$ if $z^{(k)}\le z'^{(k)}$ for every $k\in[d]$, and write $z\prec z'$ if also $z\ne z'$. We say that $z'$ \emph{dominates} $z$ when it is at least as good in every coordinate and strictly better in at least one. Under this convention, identical observations do not dominate each other. For observations $z_1,\ldots,z_t$, their \emph{Pareto frontier} is
\[
\Max(z_1,\ldots,z_t)=\{z_i: \nexists j\in [t] \text{ such that }z_i\prec z_j\}.
\]
At time $t$, we call $Z_t$ \emph{currently maximal} when $Z_t\in\Max(Z_1,\ldots,Z_t)$, and \emph{maximal} when $Z_t\in\Max(Z_1,\ldots,Z_n)$.

The decision maker knows the horizon $n$, and the observations arrive sequentially, one at a time. At time $t$, the decision maker observes $Z_t$ and either selects it and stops or rejects it permanently. The goal is to maximize the probability of selecting a maximal observation. We call this outcome a \emph{success}; selecting a nonmaximal observation or making no selection counts as a \emph{failure}.

A stopping rule $\mathbf r=(r_1,\ldots,r_n)$ consists of measurable functions $r_t:(\mathbb R^d)^t\to[0,1]$, where $r_t$ gives the stopping probability conditional on reaching time $t$. Using independent random variables $U_1,\ldots,U_n\sim\operatorname{Unif}[0,1]$, independent of the observations, its stopping time is
\[
\tau_{\mathbf r}=\inf\{t\in[n]:U_t\le r_t(Z_1,\ldots,Z_t)\},\qquad\inf\varnothing=\infty.
\]
For a fixed rule we write $\tau$ for $\tau_{\mathbf r}$. Its probability of selecting a maximal observation is
\[
p_F(\mathbf r)=\Prb_F\left\{\tau\le n,\ Z_{\tau}\in\Max(Z_1,\ldots,Z_n)\right\}.
\]
For conditioning in proofs, write
\[
\mathcal H_0=\{\varnothing,\Omega\},\qquad
\mathcal H_t=\sigma(Z_1,\ldots,Z_t,U_1,\ldots,U_t),
\quad t\in[n],
\]
where $\sigma(\cdot)$ denotes the information generated by the listed
random variables. This includes the observations and the rule's
randomization through time $t$.

For a point $z$, we let $q_F(z)=\Prb\left\{z\prec Z\right\}$ be the \emph{domination risk} of $z$, which is the probability that an observation $Z\sim F$ dominates it. We write $q(z)$ when $F$ is clear from context. Conditional on the observations through time $t$, the probability that $Z_t$ is maximal is
\begin{equation}\label{eq:conditional-success}
\Prb\left\{Z_t\in\Max(Z_1,\ldots,Z_n)\mid Z_1,\ldots,Z_t\right\}
=\ind\{Z_t\in\Max(Z_1,\ldots,Z_t)\}(1-q(Z_t))^{n-t}.
\end{equation}
The indicator requires that no earlier observation dominates $Z_t$ and the remaining factor is the probability that none of the independent
future observations does. An observation already dominated when it arrives cannot be maximal in the complete sample. A rule can reject such an observation and continue without decreasing its probability of selecting a maximal one. We may therefore restrict attention to rules that select only currently maximal observations whenever we seek an optimal rule or construct one. For any such rule, failure is the disjoint union of \[ \underbrace{\{\tau=\infty\}}_{\text{no selection}} \quad\text{and}\quad \underbrace{\{\tau\le n,\ \exists j>\tau:\ Z_\tau\prec Z_j\}}_{\text{later domination}}. \] These names refer to these two events. An unrestricted rule can also fail by selecting an observation already dominated when it arrives.

\paragraph{Additional Notation.} All logarithms are natural unless a base is displayed. Dimensions and distribution parameters are fixed as $n\to\infty$ unless stated otherwise. We use $c$ for a positive absolute constant and subscripts to show allowed dependence, as in $c_d$ and $c_\Sigma$. Such constants may change between inequalities unless explicitly fixed. Subscripts on $O$, $o$, $\Omega$, and $\Theta$ indicate the same allowed dependence. Without a subscript, dependence on fixed parameters is understood from context. We write $X_s\xrightarrow{p}X$ for convergence in probability and $o_p(1)$ for a random quantity converging to zero in probability.

\subsection{Information Settings}

We distinguish what the decision maker observes from what it knows about $F$. In the \emph{full-information} setting, it knows $F$ and observes numerical values. This is the natural multivariate extension of the full-information best-choice problem of Gilbert and Mosteller~\cite{gilbert1966}, with the objective of selecting a maximal observation. Write
\[
p_n^\star(F)=\sup_{\mathbf r}p_F(\mathbf r)
\]
for the optimal probability in this setting; the rule may depend on $F$.

A \emph{rank-based} rule uses only coordinate comparisons among the observations seen so far, together with their arrival order. Such rules are natural when the marginal distributions are unknown, as comparisons still determine which observations are maximal, without requiring a numerical scale for each criterion. Formally, for $i\le t$ and $k\in[d]$, let
\[
R_{t,i}^{(k)}=1+\#\{j\le t:Z_j^{(k)}<Z_i^{(k)}\},\qquad
\mathcal R_t=(R_{t,i}^{(k)})_{i\in[t],\,k\in[d]}.
\]
The array $\mathcal R_t$ represents the observed points on the \emph{rank grid} $[t]^d$, with label $i$ recording arrival time (Figure~\ref{fig:rank-grid}). A rule is \emph{rank-based} if its stopping probability at time $t$ depends on the observations only through $\mathcal R_t$. It may still use prior knowledge about their distribution.

\begin{figure}[tbp]
\centering
\begin{tikzpicture}[x=4.0cm,y=4.0cm,font=\small,>=Stealth]
\begin{scope}
\draw[->,ParetoGray] (0,0)--(1.08,0) node[right] {$z^{(1)}$};
\draw[->,ParetoGray] (0,0)--(0,1.08) node[above] {$z^{(2)}$};
\foreach \x/\y/\i in {.12/.78/1,.91/.92/2,.19/.12/3,.27/.19/4,.75/.25/5,.83/.84/6}{
 \draw[ParetoGray!35,densely dashed] (\x,0)--(\x,\y)--(0,\y);
 \fill[ParetoBlue] (\x,\y) circle[radius=2.4pt];
}
\node[text=black] at (.50,-.17) {(a) Numerical observations};
\end{scope}
\draw[->,ParetoOrange,thick] (1.28,.50)--(1.80,.50)
 node[midway,above,align=center,text=black,font=\footnotesize] {replace values\\by ranks};
\begin{scope}[xshift=8.4cm]
\foreach \r in {1,...,6}{
 \draw[ParetoGray!35] (.15*\r,.15)--(.15*\r,.90);
 \draw[ParetoGray!35] (.15,.15*\r)--(.90,.15*\r);
 \node[below,text=ParetoGray] at (.15*\r,0) {\r};
 \node[left,text=ParetoGray] at (0,.15*\r) {\r};
}
\draw[->,ParetoGray] (0,0)--(1.08,0) node[right] {$R^{(1)}_{6,i}$};
\draw[->,ParetoGray] (0,0)--(0,1.08) node[above] {$R^{(2)}_{6,i}$};
\foreach \x/\y/\i in {1/4/1,6/6/2,2/1/3,3/2/4,4/3/5,5/5/6}{
 \fill[ParetoBlue] (.15*\x,.15*\y) circle[radius=2.4pt];
}
\node[text=black] at (.50,-.17) {(b) The rank grid};
\end{scope}
\end{tikzpicture}
\caption{Six observations and their coordinate ranks on a $6\times6$ grid. Replacing values by ranks preserves comparisons but removes numerical distances.}
\label{fig:rank-grid}
\end{figure}
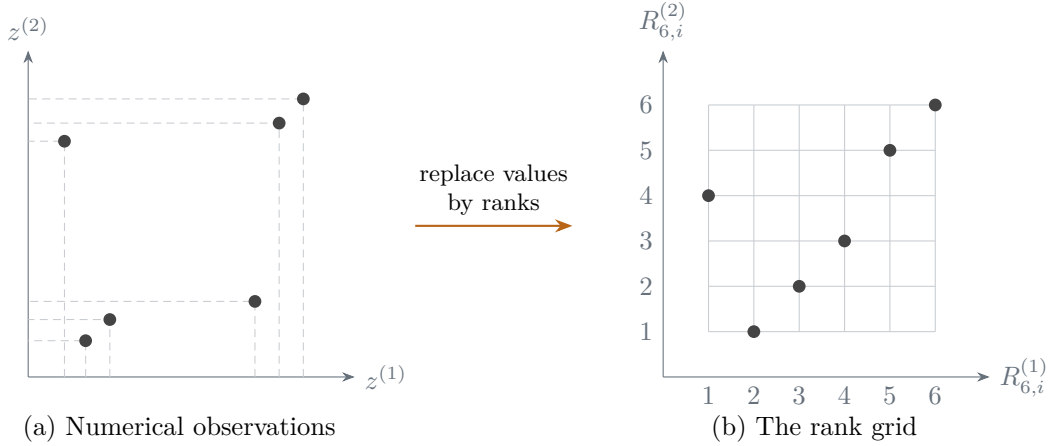

To specify that knowledge, let $G$ be the distribution of $\mathcal R_n$. It gives the probability of each arrangement of the $n$ labeled observations on the rank grid. Write
\[
\mathcal F(G)=\{F:\mathcal R_n\text{ has distribution }G
                         \text{ under }F^{\otimes n}\}.
\]
The complete array determines all earlier comparisons and whether a selected observation is maximal. Thus any fixed rank-based rule has the same value for every $F\in\mathcal F(G)$. Let $p^\star_{n,\mathrm{rank}}(G)$ be the largest such value. The following theorem shows that, if only $G$ is known, numerical observations cannot improve the worst-case guarantee.

\begin{theorem}\label{thm:rank-sufficiency}
Fix $n,d$ and a distribution $G$ of the rank array $\mathcal R_n$ that arises from i.i.d. observations. A rule may know $G$ and use numerical observations, but does not know which $F\in\mathcal F(G)$ generates them. Then
\[
 \sup_{\mathbf r}\inf_{F\in\mathcal F(G)}p_F(\mathbf r)
 =p^\star_{n,\mathrm{rank}}(G).
\]
More strongly, fix any $F_0\in\mathcal F(G)$, any numerical rule $\mathbf r$, and $\varepsilon>0$. There are strictly increasing coordinate maps $T_1,\ldots,T_d$ and a rank-based rule $\mathbf r'$ such that, for $F$ the distribution of $(T_1(Z^{(1)}),\ldots,T_d(Z^{(d)}))$ with $Z\sim F_0$,
\[
 p_F(\mathbf r)\le p_F(\mathbf r')+\varepsilon.
\]
\end{theorem}

The proof in Appendix~\ref{app:rank-sufficiency} adapts the finite-Ramsey argument of Correa et al.~\cite{correa2021}. The equality applies when the marginal scales are unknown and unrestricted. Numerical values can still help when $F$ is known.

When $F$ or $G$ is known, backward induction compares selecting now with continuing and gives an optimal rule. These decisions may depend on the full observed frontier or rank array. We next specify what is known in the two distribution families studied in this paper.

\paragraph{Independent coordinates.}
First consider independent coordinates with continuous marginal distributions. Write
\[
\mathcal F_d=\{F_1\otimes\cdots\otimes F_d:
F_1,\ldots,F_d\text{ have continuous distribution functions}\}.
\]
Continuity excludes point masses but does not require a density. In the rank-based setting, the marginals are unknown, but independence is known. It determines $G$, since the coordinate rankings are independent uniform permutations. Thus every rank-based rule has the same value throughout $\mathcal F_d$. Let $\Phi$ be the standard Gaussian distribution function,
\[
\Phi(x)=\int_{-\infty}^x\frac{e^{-t^2/2}}{\sqrt{2\pi}}\,dt.
\]
When the marginals are known, the coordinate transformation
\[
z^{(k)}\longmapsto\Phi^{-1}(F_k(z^{(k)}))
\]
gives independent standard Gaussian coordinates and preserves all sample comparisons with probability one. We denote their joint distribution by $\mathcal N_d(0,I_d)$. The endpoints of $[0,1]$ have probability zero. Thus both optimal probabilities are common to all continuous product distributions, and we can analyze them using independent standard Gaussian coordinates.

The same rank observations arise from fixed coordinate lists. For each coordinate, fix $n$ distinct values and shuffle them uniformly and independently of the other coordinates. The entries in position $i$ form observation $i$. The coordinate rankings are again independent uniform permutations, so every rank-based rule has the same probability of selecting a maximal observation as under $\mathcal F_d$.

\paragraph{Gaussian coordinates.}
More generally, write $\mathcal N_d(\mu,\Sigma)$ for the $d$-dimensional Gaussian distribution with mean vector $\mu$ and covariance matrix $\Sigma$. We also allow dependence while assuming that the observations are Gaussian. Full information specifies the mean vector and covariance matrix. A rank-based rule knows the Gaussian assumption but, unless stated otherwise, is not given these parameters. Unlike independence, the Gaussian assumption alone does not determine $G$, since the coordinate rankings depend on the correlation matrix. The rule can use the observed comparisons to estimate that matrix. For singular Gaussian distributions, we specify separately when the correlation matrix is given.

Strictly increasing transformations of individual coordinates preserve all comparisons. The Gaussian rank model therefore also covers observations obtained by applying such transformations to Gaussian coordinates, even when the transformed marginals are not Gaussian. Sections~\ref{sec:fixed-cutoffs} and~\ref{sec:optimal-stopping} analyze independent coordinates through $\mathcal N_d(0,I_d)$. Section~\ref{sec:gaussian} extends the analysis to Gaussian dependence, including singular covariances. Independent continuous coordinates are therefore included in this comparison model. In the independent model, independence is given to the rule; in the Gaussian model, the correlations may be unknown.

\section{Fixed-Cutoff Rules under Independence}
\label{sec:fixed-cutoffs}
\label{sec:d-dim-cutoff}

We begin with a fixed-cutoff rule for independent coordinates with
continuous marginal distributions. It rejects the first $s$ observations,
then selects the first currently maximal observation. Waiting longer
reduces the chance of later domination but leaves fewer opportunities
to select. We determine when its failure probability tends to zero and
how small it can be.

Throughout this section $d\ge2$ is fixed and omitted from the notation.
Write $m=n-s$. We sometimes use $s_n$ and $m_n$ to show their dependence
on the horizon $n$.

\subsection{Failure Decomposition}
\label{subsec:cutoff-rule-definition}

As explained in \Cref{sec:model}, an increasing transformation of each coordinate allows us to analyze every continuous product distribution using $F=\mathcal N_d(0,I_d)$. We use the Gaussian estimates in Appendix~\ref{app:stopping-estimates} with $\Sigma=I_d$. Throughout this section and Section~\ref{sec:optimal-stopping}, $Z_1,\ldots,Z_n$ have this distribution. The probability $q$ from Section~\ref{sec:model} is therefore
\begin{equation*}
q(z)=\prod_{k=1}^d(1-\Phi(z^{(k)})).
\end{equation*}

For $s\in\{1,\ldots,n-1\}$, the rule has stopping time
\[
\tau=\inf\{s<t\le n:Z_t\in\Max(Z_1,\ldots,Z_t)\},\qquad\inf\varnothing=\infty.
\]
Write $p_n(s)$ for its probability of selecting a maximal observation among all $n$ observations. Since the rule uses only coordinate comparisons, this probability is the same for all independent continuous marginals.

Let $\mathcal U_s$ be the set of points not dominated by the initial sample, with probability $A_s$ under $F$,
\[
\mathcal U_s=\{z\in\mathbb R^d:z\not\prec Z_t\text{ for all }t\le s\},\qquad A_s=F(\mathcal U_s).
\]
More generally, $\mathcal U_t$ and $A_t$ denote this set and its probability after $t$ observations. The rule selects the first observation after time $s$ in $\mathcal U_s$. Before that selection, every observation after the cutoff is dominated by one of the first $s$ observations, so it cannot dominate a point in $\mathcal U_s$. Figure~\ref{fig:cutoff-geometry} shows the same region after applying $\Phi$ to each coordinate, so that its probability is represented by area in the unit square.

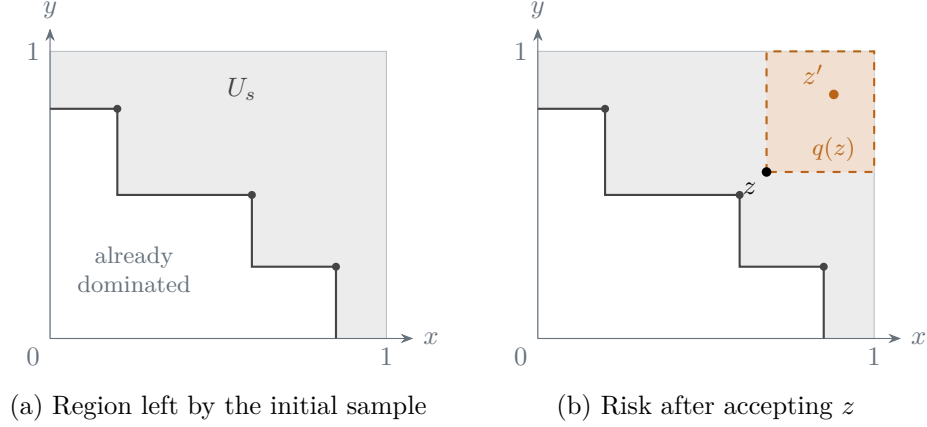
\begin{figure}[tbp]
\centering
\begin{tikzpicture}[x=4.45cm,y=3.8cm,font=\small,>=Stealth]
\foreach \shift in {0,1.45} {
 \begin{scope}[xshift=\shift*4.45cm]
  \fill[ParetoBlue!10] (0,.8)--(.2,.8)--(.2,.5)--(.6,.5)--(.6,.25)--(.85,.25)--(.85,0)--(1,0)--(1,1)--(0,1)--cycle;
  \draw[ParetoGray!45] (0,0) rectangle (1,1);
  \draw[->,ParetoGray] (0,0)--(1.08,0) node[right] {$x$};
  \draw[->,ParetoGray] (0,0)--(0,1.08) node[above] {$y$};
  \draw[thick,ParetoBlue] (0,.8)--(.2,.8)--(.2,.5)--(.6,.5)--(.6,.25)--(.85,.25)--(.85,0);
  \foreach \x/\y in {.2/.8,.6/.5,.85/.25} {\fill[ParetoBlue] (\x,\y) circle (1.5pt);}
  \node[below left,text=ParetoGray] at (0,0) {$0$};
  \node[below,text=ParetoGray] at (1,0) {$1$};
  \node[left,text=ParetoGray] at (0,1) {$1$};
 \end{scope}
}
\node[ParetoBlue] at (.57,.87) {$U_s$};
\node[align=center,text=ParetoGray,font=\footnotesize] at (.25,.24) {already\\dominated};
\node[align=center] at (.5,-.24) {(a) Region left by the initial sample};
\begin{scope}[xshift=1.45*4.45cm]
 \filldraw[fill=ParetoOrange!20,draw=ParetoOrange,dashed,thick] (.68,.58) rectangle (1,1);
 \fill[black] (.68,.58) circle (1.8pt) node[below left] {$z$};
 \fill[ParetoOrange] (.88,.85) circle (1.8pt) node[above left] {$z'$};
 \node[ParetoOrange,font=\footnotesize] at (.88,.66) {$q(z)$};
 \node[align=center] at (.5,-.24) {(b) Risk after accepting $z$};
\end{scope}
\end{tikzpicture}
\caption{The initial frontier (dark gray) bounds the remaining region $U_s$. The rule selects its first hit $z$; a later point in the orange rectangle dominates $z$.}
\label{fig:cutoff-geometry}
\end{figure}

The rule fails if it does not select any observation, or if the selected point is dominated by a later observation. Write $Z_{1:s}=(Z_1,\ldots,Z_s)$. Conditional on this initial sample, the probability of making no selection is $(1-A_s)^m$. The rule first selects at time $s+j$ with probability $(1-A_s)^{j-1}A_s$, and the selected point has distribution $F$ conditional on $\mathcal U_s$. Imagine placing the $m-j$ future observations first. They dominate the selected point exactly when it lies in the part they remove from $\mathcal U_s$. Given $Z_{1:s}$, this removed probability has the same distribution as $A_s-A_{s+m-j}$. Since these observations are independent of the selected point, the conditional probability that a later observation dominates the selected point is $\E\left[A_s-A_{s+m-j}\mid Z_{1:s}\right]/A_s$. Multiplying by the selection probability and summing gives the exact failure probability
\begin{equation}\label{eq:cutoff-failure-identity}
1-p_{s+m}(s)
=\E\left[(1-A_s)^m\right]
+\E\left[\sum_{j=1}^m(1-A_s)^{j-1}(A_s-A_{s+m-j})\right].
\end{equation}

\subsection{Vanishing Failure Probability}

The two terms in \eqref{eq:cutoff-failure-identity} admit the following approximation.

\begin{theorem}
\label{thm:cutoff-quantitative}
As $s\to\infty$, for all integers $m\ge1$,
\begin{equation}\label{eq:cutoff-general-failure}
1-p_{s+m}(s)
=e^{-m\E\left[A_s\right]}+1-\frac{\E\left[A_{s+m}\right]}{\E\left[A_s\right]}
+o_d(1).
\end{equation}
The error tends to zero independently of $m$.
\end{theorem}

\begin{proof}
Lemma~\ref{lem:cutoff-common-replacements} gives $\E\left[(1-A_s)^m\right]=e^{-m\E\left[A_s\right]}+o_d(1)$. For the second term of \eqref{eq:cutoff-failure-identity}, that lemma allows us to replace $A_{s+m-j}$ by $A_{s+m}$ and sum the geometric series,
\begin{equation}\label{eq:cutoff-later-domination}
\begin{aligned}
\E\left[\sum_{j=1}^m(1-A_s)^{j-1}(A_s-A_{s+m-j})\right]&=\E\left[\left(1-\frac{A_{s+m}}{A_s}\right)
             (1-(1-A_s)^m)\right]+O_d((\log s)^{1-d})\\
&=1-\E\left[\frac{A_{s+m}}{A_s}\right]+O_d((\log s)^{1-d})\\
&=1-\frac{\E\left[A_{s+m}\right]}{\E\left[A_s\right]}+o_d(1).
\end{aligned}
\end{equation}
The last two equalities use \eqref{eq:cutoff-volume-cross-term} and the mean-ratio bound in that lemma.
\end{proof}

For example, if $s_n/n\to\alpha\in(0,1)$, the mean estimate in
Lemma~\ref{gauss:cutoff-volume-estimate} gives
\[
m_n\E\left[A_{s_n}\right]\to\infty,\qquad
\frac{\E\left[A_n\right]}{\E\left[A_{s_n}\right]}
\sim\frac{s_n}{n}\left(\frac{\log n}{\log s_n}\right)^{d-1}
\to\alpha.
\]
Hence $p_n(s_n)\to\alpha$ by \eqref{eq:cutoff-general-failure}.
To make failure vanish, the fraction left after the cutoff must instead
tend to zero, while leaving enough observations to select. Write
\[
a_s=\frac{(\log s)^{d-1}}{(d-1)!s}\sim\E\left[A_s\right].
\]

\begin{theorem}
\label{thm:fixed-d-successful-cutoff-families}
For any sequence $s_n\in\{1,\ldots,n-1\}$ and $m_n=n-s_n$,
\[
1-p_n(s_n)\to0
\quad\Longleftrightarrow\quad
\frac{m_n}{s_n}\to0\quad\text{and}\quad m_na_{s_n}\to\infty.
\]
\end{theorem}

\begin{proof}
Under the two conditions, $s_n\sim n\to\infty$, $\E\left[A_n\right]/\E\left[A_{s_n}\right]\to1$, and $m_n\E\left[A_{s_n}\right]\to\infty$, by \Cref{gauss:cutoff-volume-estimate}. Substituting into \Cref{thm:cutoff-quantitative} gives $1-p_n(s_n)\to0$.

Conversely, suppose $1-p_n(s_n)\to0$. First, $s_n\to\infty$. Otherwise some subsequence uses a fixed cutoff. Continue the observations indefinitely. The rule eventually selects a point, and that point has $q(Z)>0$ with probability one. The probability that no later observation dominates it tends to zero as the horizon increases, a contradiction.

For any fixed $c>0$, if $m\ge cs$, then $A_{s+m}\le A_{s+\lceil cs\rceil}$. Thus \Cref{thm:cutoff-quantitative} gives
\[
1-p_{s+m}(s)
\ge1-\frac{\E\left[A_{s+\lceil cs\rceil}\right]}{\E\left[A_s\right]}
-o_d(1)
=\frac{c}{1+c}+o(1).
\]
The last equality follows from \Cref{gauss:cutoff-volume-estimate}. Thus $1-p_n(s_n)\to0$ requires $m_n/s_n\to0$. It follows that $\E\left[A_n\right]/\E\left[A_{s_n}\right]\to1$, so \Cref{thm:cutoff-quantitative} also requires $m_n\E\left[A_{s_n}\right]\to\infty$. Since $\E\left[A_{s_n}\right]\sim a_{s_n}$, this gives the other necessary condition.
\end{proof}

For $d=2$, for example, $m_n=\lceil n\log\log n/\log n\rceil$ satisfies both conditions. The rule can therefore reject all but a vanishing fraction of the observations and still select a maximal observation with probability tending to one.

\subsection{Optimal Cutoff}

The two causes of failure suggest minimizing $e^{-ma_s}+m/s$.
Writing $x=ma_s$ gives $e^{-x}+x/(sa_s)$, whose minimum occurs at
$x=\log(sa_s)$. The next theorem verifies the resulting cutoff.

\begin{theorem}\label{thm:cutoff-rate-lower-bound}
For fixed $d\ge2$,
\[
\min_{1\le s<n}(1-p_n(s))
\sim(d-1)(d-1)!\frac{\log\log n}{(\log n)^{d-1}}.
\]
Every minimizing cutoff has
$m_n\sim(d-1)(d-1)!n\log\log n/(\log n)^{d-1}$.
\end{theorem}

For all sufficiently large $n$, the leading term alone gives an
asymptotically optimal choice,
\begin{equation}\label{eq:cutoff-explicit-choice}
m_n=\left\lceil(d-1)(d-1)!\frac{n\log\log n}{(\log n)^{d-1}}\right\rceil,
\qquad s_n=n-m_n.
\end{equation}

\begin{proof}
We first sharpen the probability of later domination. Suppose
$m/s\to0$ and $ma_s\to\infty$. The first two equalities in
\eqref{eq:cutoff-later-domination} give
\begin{equation}\label{eq:cutoff-later-equivalent}
\begin{aligned}
\Prb\left\{\tau\le n,\ Z_\tau\notin\Max(Z_1,\ldots,Z_n)\right\}
&=\E\left[\frac1{A_s}\int_{\mathcal U_s}(1-(1-q(z))^m)\,dF(z)\right]
 +O_d((\log s)^{1-d})\\
&\sim m/s.
\end{aligned}
\end{equation}
For the upper bound, use $1-(1-q)^m\le mq$ on
$\{A_s\ge(1-\varepsilon)\E\left[A_s\right]\}$; for the lower bound, use
$1-(1-q)^m\ge mq-m^2q^2/2$ on $\{q\le T/s\}$.
Lemma~\ref{gauss:cutoff-volume-estimate} bounds the resulting integral
from above and below, respectively, by
\[
\frac{1+o(1)}{1-\varepsilon}\frac ms
 +O_{d,\varepsilon}((\log s)^{1-d}),
\qquad
\frac ms\left(\int_0^T te^{-t}\,dt+o(1)\right)
 -\frac{m^2T^2}{2s^2}.
\]
For fixed $T$ and $\varepsilon$, the errors are $o(m/s)$ because $ma_s\to\infty$ and $m/s\to0$, respectively.
Let $T\to\infty$ and $\varepsilon\downarrow0$.

For the choice \eqref{eq:cutoff-explicit-choice},
$ma_s=(d-1)\log\log s+o(1)$.
The exponential-moment estimate \eqref{gauss:cutoff-no-selection}
bounds $\Prb\left\{\tau=\infty\right\}$ by $O_d((\log s)^{1-d})=o(m/s)$.
Adding this to \eqref{eq:cutoff-later-equivalent} proves the upper bound.

Conversely, the upper bound and
Theorem~\ref{thm:fixed-d-successful-cutoff-families} imply $s\sim n$
and $ma_s\to\infty$ for every minimizing cutoff. Jensen's inequality gives
\begin{equation}\label{eq:cutoff-no-choice-bounds}
(1-\E\left[A_s\right])^m\le\Prb\left\{\tau=\infty\right\}\le1-p_n(s)
=O_d\left(\frac{\log\log n}{(\log n)^{d-1}}\right).
\end{equation}
Taking logarithms shows that
$m\E\left[A_s\right]\ge(d-1-o(1))\log\log n$. Consequently,
\[
1-p_n(s)\ge(1-o(1))m/s
\ge\bigl((d-1)(d-1)!-o(1)\bigr)
\frac{\log\log n}{(\log n)^{d-1}}.
\]
The matching upper bound fixes both the leading failure probability
and the asymptotic location of every minimizing cutoff.
\end{proof}

For $d=2$, the best cutoff has failure probability asymptotic to
$\log\log n/\log n$. Section~\ref{sec:optimal-stopping} improves this
by a factor asymptotic to $\tfrac12(\log\log n)^2$ using a threshold
that changes with the number of observations remaining.

\section{Optimal Stopping under Independence}
\label{sec:optimal-stopping}

Fixed-cutoff rules can select a maximal observation with probability tending to one, but their failure probability can be reduced further. We now let the threshold depend on how many observations remain. This gives the leading constant in the smallest possible failure probability. A rule with known marginals and a rule using only coordinate ranks both achieve it.

As in Section~\ref{sec:fixed-cutoffs}, we work with independent standard Gaussian coordinates with $F=\mathcal N_d(0,I_d)$, write $p(\mathbf r)=p_F(\mathbf r)$, and put $m=n-s$. Let $G$ be the distribution of their rank array, as in Section~\ref{sec:model}.

\begin{theorem}\label{opt:thm:optimal-rate}
For every fixed $d\ge2$,
\begin{equation}\label{opt:eq:optimal-rate}
1-p^\star_{n,\mathrm{rank}}(G)
\sim1-p_n^\star(F)
\sim\frac{2(d-2)!}{(\log n)^{d-1}\log\log n}.
\end{equation}
For every product distribution with continuous marginals, the two time-dependent rules below achieve this asymptotic. One uses the known marginals; the other uses only coordinate ranks.
\end{theorem}

By \Cref{thm:cutoff-rate-lower-bound}, the ratio of the best fixed-cutoff failure probability to either optimal failure probability is asymptotic to $\tfrac12(d-1)^2(\log\log n)^2$.

The lower bound rests on the fact that points with very small $q$ are rare, so a rule that selects with high probability cannot always avoid the chance of later domination. For the upper bound, we compare $q(Z_s)$ with a threshold that depends on $m$. We first analyze this rule when the marginals are known, then make the same decision using coordinate ranks. For both rules, we bound the probability of reaching each decision and sum its contribution to failure.

\subsection{Lower Bound}
\label{sec:clairvoyant-benchmark}

Selecting a currently maximal observation at time $s$ has conditional failure probability $1-(1-q(Z_s))^{n-s}$, by \eqref{eq:conditional-success}. When this probability is small, it is close to $(n-s)q(Z_s)$. A rule that rarely fails must therefore select points with small $q$, while still selecting with high probability. The distribution
\begin{equation}\label{eq:product-risk-cdf}
H(v):=\Prb\left\{q(Z)\le v\right\}
=v\sum_{j=0}^{d-1}\frac{(\log(1/v))^j}{j!},\qquad 0<v\le1,
\end{equation}
limits how often such points arrive. The formula is the tail probability of the sum $-\log q(Z)$ of $d$ independent mean-one exponential variables. We use it to bound the expected $q$ at a selection, then sum the resulting contributions to failure.

\begin{theorem}\label{thm:general-d-fi-lower-bound}
For every fixed $d\ge2$, every full-information stopping rule satisfies
\[
1-p(\mathbf r)
\ge\frac{2(d-2)!-o_d(1)}{(\log n)^{d-1}\log\log n}.
\]
\end{theorem}

\begin{proof}
It suffices to consider rules with failure probability at most
$3(d-2)!/[(\log n)^{d-1}\log\log n]$.
By the reduction in Section~\ref{sec:model}, assume the rule selects
only currently maximal observations. Then
\begin{equation}\label{eq:failure-decomposition}
1-p(\mathbf r)=\Prb\left\{\tau=\infty\right\}
+\sum_{s=1}^n\E\left[
\ind\{\tau=s\}\bigl(1-(1-q(Z_s))^{n-s}\bigr)\right].
\end{equation}
\proofstep{Choices used in the bound.}
Put $\ell=\lfloor n/[(\log n)^{d-1}\log\log n]\rfloor$.
For $s\le n-\ell$, define
\begin{equation}\label{eq:sharp-lower-q-range}
\mathcal B_s=\left\{z\in\mathbb R^d:
\frac1{n(\log n)^{3(d-1)}}\le q(z)
\le\frac1{(n-s)\log\log n}\right\}.
\end{equation}
The lower boundary excludes points that almost never appear; the
upper boundary makes the probability of later domination asymptotic to $(n-s)q(z)$.
Set $\tau'=\tau$ if $\tau\le n-\ell$ and $Z_\tau\in\mathcal B_\tau$,
and $\tau'=\infty$ otherwise. Then
\begin{equation}\label{eq:sharp-lower-retained-mass}
\begin{aligned}
\Prb\left\{\tau'=\infty\right\}
&\le(1-p(\mathbf r))
 +\sum_{s=n-\ell+1}^n\E\left[(1-q(Z_s))^{s-1}\right]
 +nH\left(\frac1{n(\log n)^{3(d-1)}}\right)
 +\frac{1-p(\mathbf r)}{1-e^{-1/\log\log n}}\\
&=O((1-p(\mathbf r))\log\log n)+o(1)=o(1).
\end{aligned}
\end{equation}
The four terms cover $\tau=\infty$, a choice in the final $\ell$ observations, and choices below and above the two boundaries of $\mathcal B_s$.
The second term is $O(\ell(\log n)^{d-1}/n)=o(1)$ by
\eqref{gauss:laplace-estimates}. For the last term, each such choice
has conditional probability at least $1-e^{-1/\log\log n}$ of later
domination, so its probability times this factor cannot exceed
$1-p(\mathbf r)$.

\proofstep{Failure at a stop.}
Let $Z$ be an independent observation. The definition of $\mathcal B_s$
and \eqref{eq:product-risk-cdf} give, for every $x\ge0$ and $s\le n-\ell$,
\begin{equation}\label{eq:sharp-lower-cdf-envelope}
\Prb\left\{Z\in\mathcal B_s,\ q(Z)\le x\right\}
\le(1+o(1))\frac{(\log n)^{d-1}}{(d-1)!}x.
\end{equation}
Indeed, $\log(1/q(z))=\log n+O(\log\log n)$ on $\mathcal B_s$;
below its lower boundary the probability is zero, and above its upper
boundary it is bounded by $H$ at that boundary.
Conditioning on the information $\mathcal H_{s-1}$ from Section~\ref{sec:model}, independence gives
\[
\Prb\left\{\tau'=s,\ q(Z_s)\le x\mid\mathcal H_{s-1}\right\}
\le\Prb\left\{Z\in\mathcal B_s,\ q(Z)\le x\right\}.
\]
Also, since $\{\tau'\ge s\}$ is determined before $Z_s$ arrives and
$\mathcal B_s\subseteq\mathcal B_{n-\ell}$,
\[
\frac{\Prb\left\{\tau'=s\right\}}{\Prb\left\{\tau'\ge s\right\}}
\le\Prb\left\{Z\in\mathcal B_s\right\}
\le\Prb\left\{Z\in\mathcal B_{n-\ell}\right\}
=O\left(\frac{(\log n)^{2(d-1)}}n\right)=o(1).
\]
Integrating the preceding conditional bound and applying
Cauchy--Schwarz gives
\begin{equation}\label{eq:sharp-lower-energy}
\begin{aligned}
1-p(\mathbf r)
&\ge(1-o(1))\sum_{s=1}^{n-\ell}(n-s)
       \E\left[q(Z_s)\ind\{\tau'=s\}\right]\\
&\ge\frac{(1-o(1))(d-1)!}{2(\log n)^{d-1}}
\sum_{s=1}^{n-\ell}
\frac{(n-s)\Prb\left\{\tau'=s\right\}^2}{\Prb\left\{\tau'\ge s\right\}}\\
&\ge\frac{(1-o(1))(d-1)!}{2(\log n)^{d-1}}
\frac{4\left(1-\sqrt{\Prb\left\{\tau'=\infty\right\}}\right)^2}
{\sum_{j=\ell}^{n-1}j^{-1}}\\
&=\frac{(2-o(1))(d-2)!}{(\log n)^{d-1}\log\log n}.
\end{aligned}
\end{equation}
The first line uses \eqref{eq:failure-decomposition} and
$mq-m^2q^2/2\le1-(1-q)^m\le mq$, since
$m q(Z_s)\le1/\log\log n$ on $\mathcal B_s$.
For the second, integrate the positive part of
$\Prb\left\{\tau'=s\mid\mathcal H_{s-1}\right\}-(1+o(1))(\log n)^{d-1}x/(d-1)!$
over $x$, then apply Cauchy--Schwarz on $\{\tau'\ge s\}$, with $0/0=0$.
The third uses Cauchy--Schwarz over $s$ and telescoping, since the
relative decrease above is $o(1)$. Finally, use
\eqref{eq:sharp-lower-retained-mass} and
$\sum_{j=\ell}^{n-1}j^{-1}\sim(d-1)\log\log n$.
All errors are uniform over rules. The marginal transformation in Section~\ref{sec:model} extends the result to every continuous product distribution.
\end{proof}

\subsection{Time-Dependent Threshold Rule}
\label{subsec:product-time-threshold}

In one-dimensional full-information best choice, the optimal value
threshold falls as time runs out \cite{gilbert1966}. We use the same
idea by allowing a larger probability $q(Z_s)$ when fewer observations
remain. First suppose that the continuous marginals are known.

For $n\ge3$ and $0\le m<n$, set
\begin{equation}\label{eq:threshold-smooth-rule}
\begin{aligned}
h=\frac{n\exp(\sqrt{\log\log n})}{(\log n)^{d-1}},
\qquad u_{m}=
\frac{2(d-1)!}
{(m+1)(\log n)^{d-1}
\log\left(1+\frac{m+1}{h}\right)}
\end{aligned}
\end{equation}
The full-information rule stops at
\[
\tau=\inf\{s\in[n]:Z_s\in\Max(Z_1,\ldots,Z_s),\
q(Z_s)\le u_{n-s}\},\qquad\inf\varnothing=\infty,
\]
and selects $Z_\tau$ when $\tau\le n$.

For $0<u<1$, in one dimension $q(z)\le u$ is the value threshold
$z\ge\Phi^{-1}(1-u)$. In two dimensions, uniform coordinates give
the boundary $(1-x)(1-y)=u_{n-s}$ when $0<u_{n-s}<1$.
For $u_{n-s}\ge1$, every point satisfies the threshold condition.
As the threshold increases, more points are allowed, but points
dominated by earlier observations are still excluded
(Figure~\ref{fig:product-threshold}). The increase near the end
provides enough chances to select. We prove that this rule matches
the leading constant in the lower bound.

\begin{figure}[tbp]
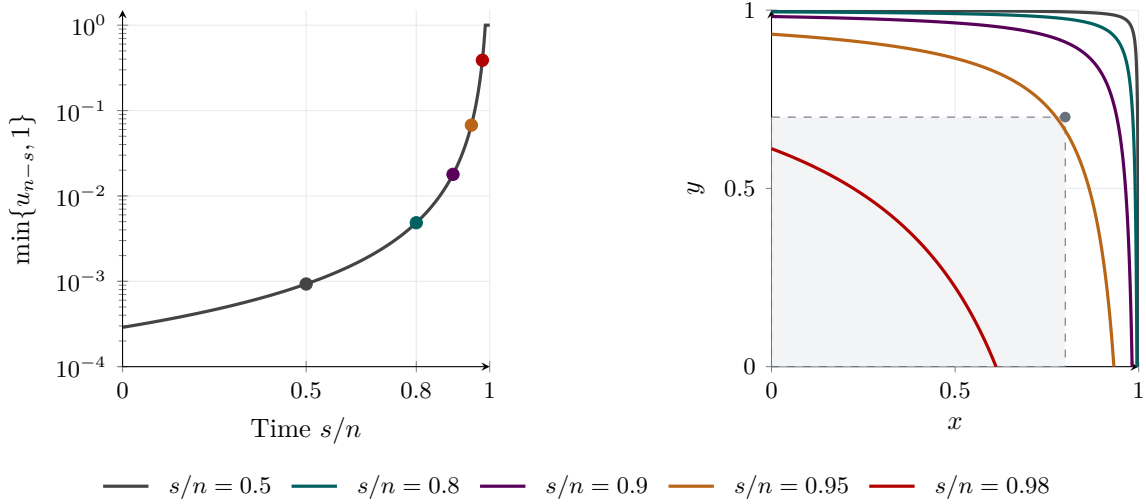

\centering

\centering


\smallskip\noindent\makebox[\textwidth][c]{\pgfplotslegendfromname{product-threshold-legend}}

\caption{Time-dependent thresholds for $n=1000$ in two dimensions, drawn in coordinates $x=\Phi(z^{(1)})$, $y=\Phi(z^{(2)})$. The left panel shows $u_{n-s}$, capped at one for display. In the right panel, the five marked thresholds give boundaries $(1-x)(1-y)=u_{n-s}$; points above and to the right qualify. The gray rectangle is excluded by an earlier observation at $(0.8,0.7)$.}
\label{fig:product-threshold}
\end{figure}

When the marginals are unknown, we estimate their upper-tail probabilities from the first $\lfloor n/2\rfloor$ observations, which the rule rejects. For a sample of size $N$ and $0\le j<N$, define
\begin{equation*}
\Prb\left\{\operatorname{Bin}(N,\overline p_N(j))\le j\right\}=n^{-2},
\qquad \overline p_N(N)=1.
\end{equation*}
This table depends only on $n$. Given $j$ observations exceeding a coordinate value, it gives an upper estimate of the probability of exceeding that value. For a later point $z$, put
\begin{equation}\label{eq:threshold-rank-score}
\widehat q(z)=\prod_{k=1}^d
\overline p_{\lfloor n/2\rfloor}
\bigl(\#\{i\le\lfloor n/2\rfloor:Z_i^{(k)}>z^{(k)}\}\bigr).
\end{equation}
The rank-based rule uses the stopping time
\begin{equation}\label{eq:estimated-threshold-rule}
\tau=\inf\{s>\lfloor n/2\rfloor: s\le n,\
Z_s\in\Max(Z_1,\ldots,Z_s),\ \widehat q(Z_s)\le u_{n-s}\},
\qquad\inf\varnothing=\infty.
\end{equation}
The initial observations remain part of the maximality comparison. The counts use only coordinate ranks, and the estimates stay fixed after the initial sample.

With probability $1-O_d(n^{-1})$, $q(z)\le\widehat q(z)$ for every $z$,
so a selection has true $q(Z_s)\le u_{n-s}$. The appendix also shows
that estimating $q$ excludes only a vanishing fraction of the points
available for selection.

\Needspace{8\baselineskip}
\subsection{Upper Bound}

\begin{theorem}\label{thm:threshold-upper-bound}
For every fixed $d\ge2$, each of the two time-dependent rules in \Cref{subsec:product-time-threshold} satisfies
\[
1-p_F(\mathbf r)
\le\frac{2(d-2)!+o_d(1)}{(\log n)^{d-1}\log\log n}
\]
for every product distribution $F$ with continuous marginals. The error term is uniform over $F$.
\end{theorem}

\begin{proof}
We work with $F=\mathcal N_d(0,I_d)$ and sum the contributions to
failure over time. Each contribution is weighted by the probability
that the rule has not yet stopped; this probability also bounds all
failure from the final few observations.

\proofstep{Selection and continuation.}
Fix $0<\varepsilon<1/4$ and put
\[
\ell=\left\lceil h\exp\left(-\tfrac14\sqrt{\log\log n}\right)\right\rceil.
\]
We analyze decisions with $\ell\le m\le\lceil n/2\rceil-1$;
all failure after these decisions will be bounded by the probability
of reaching time $n-\ell$.
Substituting the endpoints into \eqref{eq:threshold-smooth-rule}, the smallest threshold has order
$1/[n(\log n)^{d-1}\log\log n]$, and the largest has order
$e^{-\sqrt{\log\log n}/2}/n$. Thus they lie strictly within the range
of Corollary~\ref{lem:threshold-common-inputs}, which gives
\begin{equation}\label{eq:threshold-conditional-inputs}
\begin{aligned}
\Prb\left\{\tau=s\mid\mathcal H_{s-1}\right\}
&\ge(1-\varepsilon)\frac{(\log n)^{d-1}}{(d-1)!}u_m,\\
\E\left[q(Z_s)\ind\{\tau=s\}\mid\mathcal H_{s-1}\right]
&\le(1+\varepsilon)\frac{(\log n)^{d-1}}{2(d-1)!}u_m^2.
\end{aligned}
\end{equation}
for both rules on $\{\tau\ge s\}$, except at histories whose total
probability is $o(1/[(\log n)^{d-1}\log\log n])$.
Let $S_m$ be the probability of reaching time $n-m$ without selecting
or encountering such a history. Each exception is determined by the observations already seen,
so these continuation probabilities use no future information.
We include the probability of any exception in the final error term. Multiplying the continuation bounds gives
\begin{equation}\label{eq:threshold-reach-bounds}
\begin{aligned}
S_m
&\le\exp\left\{-2(1-\varepsilon)
\sum_{j=m+1}^{\lceil n/2\rceil-1}
\frac1{(j+1)\log(1+(j+1)/h)}\right\}\\
&\le\begin{cases}
(1+o(1))\left(\dfrac{\log(1+m/h)}{\log(n/h)}\right)^{2-2\varepsilon},
 &m\ge h,\\[6pt]
\dfrac{c_\varepsilon e^{-h/(2m)}}{[\log(n/h)]^{2-2\varepsilon}},
 &\ell\le m<h.
\end{cases}
\end{aligned}
\end{equation}
The integral comparison for $m\ge h$ uses $1/[x\log(1+x/h)]\ge(d/dx)\log\log(1+x/h)$.
Below $h$, use $\log(1+x/h)\le x/h$ and integrate $h/x^2$.
Endpoint errors are $o(1)$ since $\ell\to\infty$ and $h/\ell^2\to0$.
In particular, $h/\ell\sim e^{\sqrt{\log\log n}/4}$ implies
$S_\ell=o(1/[(\log n)^{d-1}\log\log n])$.

\proofstep{Summing the failure probabilities.}
The failure decomposition, $1-(1-q)^m\le mq$, and
\eqref{eq:threshold-conditional-inputs} give
\begin{equation}\label{eq:threshold-time-conclusion}
\begin{aligned}
1-p_F(\mathbf r)
&\le S_\ell+
 \sum_{s\le\lfloor n/2\rfloor}n\E\left[q(Z_s)\ind\{q(Z_s)\le u_{n-s}\}\right]\\
&\quad+\frac{(1+\varepsilon)(\log n)^{d-1}}{2(d-1)!}
 \sum_{m=\ell}^{\lceil n/2\rceil-1}S_m m u_m^2
 +o\!\left(\frac1{(\log n)^{d-1}\log\log n}\right)\\
&\le\frac{2(1+\varepsilon)(d-1)!}{(\log n)^{d-1}}
 \sum_{m=\ell}^{\lceil n/2\rceil-1}
 \frac{S_m}{m[\log(1+m/h)]^2}
 +o\!\left(\frac1{(\log n)^{d-1}\log\log n}\right)\\
&\le\frac{(2(1+\varepsilon)+o(1))(d-2)!}
 {(1-2\varepsilon)(\log n)^{d-1}\log\log n}.
\end{aligned}
\end{equation}
The first bound uses \eqref{eq:failure-decomposition} to separate
failure from time $n-\ell$ onward, initial selections, and intermediate
decisions; its error covers exceptional histories. The initial sum is absent for the rank-based rule and is
$O_d(1/[(\log n)^{d-1}(\log\log n)^2])$ for the full-information rule
by \eqref{gauss:patch-sharp-moments}. This and the bound on $S_\ell$
justify the second inequality after substituting $u_m$.
For the last, apply \eqref{eq:threshold-reach-bounds} and integral
comparison, splitting at $m=h$. The terms with $m<h$ are of smaller
order, and $\log(n/h)\sim(d-1)\log\log n$.
Let $\varepsilon\downarrow0$. Increasing marginal transformations
preserve both rules' decisions, so the bounds hold uniformly over all
continuous product distributions.
\end{proof}

\begin{proof}[Proof of \Cref{opt:thm:optimal-rate}]
\Cref{thm:general-d-fi-lower-bound} gives the lower bound on the failure probability of every full-information rule, and hence of every rank-based rule. \Cref{thm:threshold-upper-bound} gives the matching upper bound for a rule of each type. Together they imply \eqref{opt:eq:optimal-rate}, and both time-dependent rules achieve it.
\end{proof}

\section{Arbitrary Dependence}\label{sec:arbitrary-dependence}

Independent coordinates allow the probability of selecting a Pareto-maximal point to approach one. A joint density alone does not imply this result, even when every marginal is uniform. We construct a fixed dependent distribution for which the optimal probability of selecting a maximal observation with full information approaches the one-dimensional optimal probability.

Let $v_n$ denote the optimal probability of selecting the maximum of $n$ i.i.d.\ uniform observations on $(0,1)$ when their values are observed. The same value applies to every known continuous one-dimensional distribution. This is the full-information problem of Gilbert and Mosteller~\cite{gilbert1966}; its limiting value is
\begin{equation}\label{dep:eq:classical-limit}
 v_n\longrightarrow v_\infty=0.580164\ldots;
\end{equation}
see Gnedin, Kozie\l{}, and Sulkowska~\cite[Section~3]{gnedin2023} for a modern derivation. We use the full-information value $p_n^\star(F)$ from Section~\ref{sec:model}, allowing private randomization in the stopping rule.

\begin{theorem}\label{dep:thm:main}
Fix $d\ge2$. For every distribution $F$ having a joint density on $\mathbb R^d$ and every $n\ge1$,
\begin{equation}\label{dep:eq:universal-lower}
 p_n^\star(F)\ge v_n.
\end{equation}
There exists a fixed distribution $F_d$ on $(0,1)^d$, with a joint density and uniform marginals, such that, for every $n\ge2$,
\begin{equation}\label{dep:eq:sandwich}
 v_n\le p_n^\star(F_d)\le v_n+\frac4{n+1}.
\end{equation}
Thus, $p_n^\star(F_d)\to v_\infty$.
\end{theorem}

The theorem identifies $v_\infty$ as the sharp asymptotic worst-case constant over fixed distributions having a joint density, including those with uniform marginals. The distribution in this example does not depend on the horizon. Its support consists of cubes ordered so that every point in a higher cube dominates every point in a lower cube. With probability tending to one, the highest cube containing any observations contains exactly one. Selecting a maximal point then means selecting the largest first coordinate. The proof bounds the probability that this cube contains more than one observation and shows that the other coordinates do not help identify the largest first coordinate.

\subsection{A Dependent Distribution}
\label{dep:sec:construction}

We partition $(0,1)$ into successive intervals approaching $1$, apart from their endpoints. Each interval has length equal to the probability assigned to the corresponding cube; this choice makes every marginal uniform.

\begin{construction}\label{dep:construction}
For each integer $k\ge1$, set
\[
 I_k=\left(1-\frac1k,\,1-\frac1{k+1}\right),\qquad
 Q_k=I_k^d.
\]
Draw an index $J$ with $\Prb\left\{J=k\right\}=|I_k|=1/[k(k+1)]$ and, conditional on $J=k$, draw the $d$ coordinates independently and uniformly from $I_k$. Denote the resulting distribution by $F_d$.
\end{construction}

Conditional on $J=k$, the observation is uniform on $Q_k$, so $F_d$ has density
\begin{equation*}
 g_d(z)=\sum_{k=1}^\infty |I_k|^{1-d}\ind\{z\in Q_k\}.
\end{equation*}
Integrating over all but any one coordinate gives marginal density
$|I_k|^{1-d}|I_k|^{d-1}=1$ on $I_k$. Hence every marginal is uniform.

If $\ell>k$, every coordinate of every point in $Q_\ell$ exceeds the corresponding coordinate of every point in $Q_k$. Thus the cubes are ordered so that every point in a higher cube dominates every point in a lower cube, although points within a cube can be incomparable. Figure~\ref{dep:fig:cubes} illustrates the construction in dimension two. For example, $Q_1=(0,1/2)^2$ has probability $1/2$ and $Q_2=(1/2,2/3)^2$ has probability $1/6$; every observation in $Q_2$ dominates every observation in $Q_1$. The same ordering holds in every dimension.

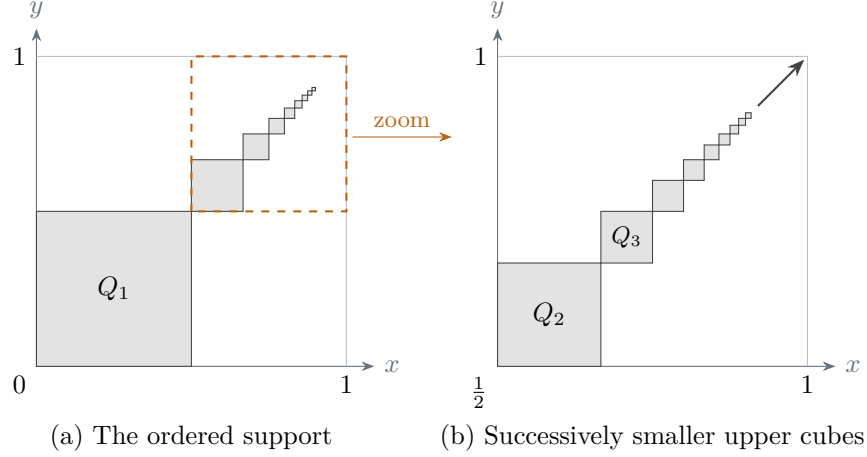
\begin{figure}[tbp]
\centering
\begin{tikzpicture}[x=4.1cm,y=4.1cm,font=\small,>=Stealth]
\draw[->,ParetoGray] (0,0)--(1.09,0) node[right] {$x$};
\draw[->,ParetoGray] (0,0)--(0,1.09) node[above] {$y$};
\draw[ParetoGray!45] (1,0)--(1,1)--(0,1);
\foreach \k in {1,...,9} {
 \pgfmathsetmacro{\lo}{1-1/\k}
 \pgfmathsetmacro{\hi}{1-1/(\k+1)}
 \filldraw[fill=ParetoBlue!15,draw=ParetoBlue] (\lo,\lo) rectangle (\hi,\hi);
}
\node at (.25,.25) {$Q_1$};
\draw[ParetoOrange,dashed,thick] (.5,.5) rectangle (1,1);
\node[below left] at (0,0) {$0$};
\node[below] at (1,0) {$1$};
\node[left] at (0,1) {$1$};
\node at (.5,-.23) {(a) The ordered support};
\draw[->,ParetoOrange] (1.02,.74)--(1.34,.74) node[midway,above,font=\footnotesize] {zoom};
\begin{scope}[xshift=6.1cm]
 \draw[->,ParetoGray] (0,0)--(1.09,0) node[right] {$x$};
 \draw[->,ParetoGray] (0,0)--(0,1.09) node[above] {$y$};
 \draw[ParetoGray!45] (1,0)--(1,1)--(0,1);
 \foreach \k in {2,...,10} {
  \pgfmathsetmacro{\lo}{2*(1-1/\k-.5)}
  \pgfmathsetmacro{\hi}{2*(1-1/(\k+1)-.5)}
  \filldraw[fill=ParetoBlue!15,draw=ParetoBlue] (\lo,\lo) rectangle (\hi,\hi);
 }
 \node at (1/6,1/6) {$Q_2$};
 \node[font=\footnotesize] at (5/12,5/12) {$Q_3$};
 \draw[->,ParetoBlue,thick] (.84,.84)--(.99,.99);
 \node[below left] at (0,0) {$\tfrac12$};
 \node[below] at (1,0) {$1$};
 \node[left] at (0,1) {$1$};
 \node at (.5,-.23) {(b) Successively smaller upper cubes};
\end{scope}
\end{tikzpicture}
\caption{The two-dimensional support, with a magnified view. Every point in a higher square dominates every point in a lower square. The squares continue toward $(1,1)$.}
\label{dep:fig:cubes}
\end{figure}

The probabilities satisfy $\Prb\left\{J\ge k\right\}=1/k$, whereas an individual cube has mass $|I_k|=1/[k(k+1)]$. Far along the chain, the probability of a particular cube is small compared with the mass above it. This makes two observations in the highest occupied cube unlikely.

\subsection{One-Dimensional Comparison}
\label{dep:sec:comparison}

Let $Z_1,\ldots,Z_n$ be i.i.d.\ from $F_d$, and write $X_i=Z_i^{(1)}$ for the first coordinate. These first coordinates are i.i.d.\ uniform on $(0,1)$. Let $T_n$ be the event that the highest occupied cube contains at least two observations. On $T_n^c$, its only observation dominates all the others and is also the unique first-coordinate maximum. Thus selection of a maximal vector on $T_n^c$ must select the one-dimensional maximum.

\begin{lemma}\label{dep:lem:collision}
For every $n\ge2$, the ordered-cube sample satisfies
\[
 \Prb\left\{T_n\right\}\le\frac4{n+1}.
\]
\end{lemma}

\begin{proof}
The event $T_n$ occurs exactly when the two largest first coordinates lie in the same interval. Let $U=\max_{i\le n}X_i$, and for $u\in I_k$ let $a(u)=1-1/k$ be the lower endpoint of that interval. Conditional on $U=u$ and on which observation has this maximum, the other $n-1$ observations are independent and uniform on $(0,u)$. Since $u-a(u)\le1/[k(k+1)]\le2(1-u)^2$ for $u\in I_k$,
\[
\begin{aligned}
\Prb\left\{T_n\right\}
&=n\int_0^1\left[1-\left(\frac{a(u)}u\right)^{n-1}\right]u^{n-1}\,du\\
&\le n(n-1)\int_0^1(u-a(u))u^{n-2}\,du\\
&\le2n(n-1)\int_0^1(1-u)^2u^{n-2}\,du
=\frac4{n+1}.
\end{aligned}
\]
The equality conditions on the maximum; the first inequality uses
$1-x^{n-1}\le(n-1)(1-x)$.
\end{proof}

A vector rule observes more than the first coordinate. For the distribution $F_d$, however, the remaining coordinates can be generated from the first coordinate as each observation arrives. This gives a one-dimensional rule that makes exactly the same decisions, and lets us use the one-dimensional upper bound. A related comparison appears in the one-dimensional problem with ties. Gnedin, Kozie\l{}, and Sulkowska~\cite[Proposition~2]{gnedin2023} bound the gain over continuous best choice by the probability of a tie for first place. Here the event to exclude is that the highest occupied cube contains at least two observations. The simulation also reproduces the coordinates within that cube.

\begin{lemma}\label{dep:lem:simulation}
For every $n\ge1$ and every possibly randomized stopping rule $\tau$ under $F_d$,
\[
 \Prb\left\{\tau\le n,\ X_\tau=\max_{i\le n}X_i\right\}\le v_n.
\]
\end{lemma}

\begin{proof}
Given an arriving one-dimensional observation $X_i=x\in I_k$, generate $d-1$ independent uniform values in $I_k$ using fresh private randomness. Together with $x$, these values have exactly the conditional distribution of $Z_i$ given its first coordinate. Feed the simulated vector to $\tau$, using an independent random seed to reproduce any private randomization in that rule.

This is a valid one-dimensional stopping rule with the same joint distribution of observations and decisions as the vector rule. Its probability of selecting the largest first coordinate is therefore at most $v_n$, since independent randomization cannot improve the one-dimensional optimum.
\end{proof}

\begin{proof}[Proof of Theorem~\ref{dep:thm:main}]
For a distribution $F$ having a joint density, the first marginal is continuous. A stopping rule can ignore the other coordinates and apply the optimal one-dimensional rule for this known marginal. With probability $v_n$, it selects the unique first-coordinate maximum. That observation is necessarily Pareto maximal, proving~\eqref{dep:eq:universal-lower}.

Construction~\ref{dep:construction} gives a single distribution $F_d$ with a joint density and uniform marginals, independently of $n$. Outside $T_n$, a Pareto-maximal selection must be the first-coordinate maximum. Thus every vector stopping rule $\tau$ satisfies
\begin{align*}
 &\Prb\left\{\tau\le n,\ Z_\tau\in\Max(Z_1,\ldots,Z_n)\right\}\\
 &\qquad\le
   \Prb\left\{\tau\le n,\ X_\tau=\max_{i\le n}X_i\right\}
   +\Prb\left\{T_n\right\}
 \le v_n+\frac4{n+1},
\end{align*}
by Lemmas~\ref{dep:lem:collision} and~\ref{dep:lem:simulation}. Taking the supremum over all rules proves the upper bound in~\eqref{dep:eq:sandwich}; its lower bound follows from~\eqref{dep:eq:universal-lower}. Finally, \eqref{dep:eq:classical-limit} gives $p_n^\star(F_d)\to v_\infty$.
\end{proof}

Within each cube, the distribution has a $d$-dimensional density. The cubes are nevertheless ordered so that every point in a higher cube dominates every point in a lower one. With probability $1-O(1/n)$, the highest occupied cube contains exactly one observation, so there is only one maximal observation. A positive definite covariance matrix does not prevent this behavior. Indeed, if the variance in a nonzero direction were zero, the distribution would lie on a hyperplane, which is impossible for a distribution with a joint density. Since $F_d$ has bounded support, its covariance matrix exists and is positive definite.

\section{Gaussian Dependence}\label{sec:gaussian}
\label{gauss:cone-stopping-section}

Arbitrary dependence can keep the optimal selection probability below
one, as Section~\ref{sec:arbitrary-dependence} shows. We now extend
the results under independence to Gaussian dependence. We first
transform the observations into independent standard Gaussian variables,
so that the region dominating each point becomes a cone with its vertex
at that point (Section~\ref{gauss:cone-representation}). We then extend
the fixed-cutoff and optimal failure results from independent coordinates,
with the rank-based rule also estimating correlations
(Sections~\ref{gauss:fixed-cutoffs} and~\ref{gauss:optimal-stopping}).
We next determine the optimal failure asymptotics for singular Gaussian
distributions, including cases where every observation is maximal or
the problem reduces to one-dimensional selection
(Section~\ref{gauss:singular-gaussian}). Finally, we relate the failure
probabilities to the expected frontier size
(Section~\ref{gauss:comparison-independent}).

Throughout the first three subsections, let $F=\mathcal N_d(0,\Sigma)$,
where $d\ge2$ and the positive definite correlation matrix $\Sigma$
are fixed. Subtracting each coordinate's mean and dividing by its
standard deviation preserves all coordinate comparisons, so this
includes every nonsingular Gaussian distribution.
We continue to use $q$, $\mathcal U_s$, $A_s=F(\mathcal U_s)$, and
$m=n-s$ as before. We also use $p_n(s)$ and $p(\mathbf r)$ for the
probabilities of selecting a maximal observation under the present
distribution $F$.

\subsection{Cone Representation}
\label{gauss:cone-representation}

In two dimensions, a linear change of coordinates turns the elliptical Gaussian density contours into circles. The same change turns the upper-right region above each point into a wedge. Thus the distribution becomes standard Gaussian, but the comparison between points changes. In dimension $d$, write
\[
X_s=\Sigma^{-1/2}Z_s\sim \mathcal N_d(0,I_d),\qquad
C=\{\Sigma^{-1/2}y:y\in\mathbb R_+^d\}.
\]
Here $\mathbb R_+^d=[0,\infty)^d$, $\Sigma^{1/2}$ is the symmetric
positive definite square root of $\Sigma$, and $\Sigma^{-1/2}$
is its inverse. The cone $C$ contains all transformed displacements that do not decrease any original coordinate. Following Section~\ref{sec:model}, define
\[
x\preceq_C x'\quad\Longleftrightarrow\quad x'-x\in C,
\qquad
x\prec_C x'\quad\Longleftrightarrow\quad
x\preceq_C x'\ \text{and }x\ne x'.
\]
Then $Z_s\prec Z_i$ if and only if $X_s\prec_C X_i$. For $x=\Sigma^{-1/2}z$ and an independent observation $Z=\Sigma^{1/2}X\sim F$, define
\[
\begin{aligned}
q_C(x)&:=\Prb\left\{x\prec_C X\right\}=\Prb\left\{X\in x+C\right\}\\
&=(2\pi)^{-d/2}\int_C e^{-\|x+y\|^2/2}\,dy=q_F(z).
\end{aligned}
\]
Here $\|\cdot\|$ is the Euclidean norm, $x+C=\{x+y:y\in C\}$, and including $x$ does not change the probability because $\Prb\left\{X=x\right\}=0$. We use $q_C$ for estimates in these coordinates and keep $q=q_F$ for the original observations. This change preserves maximality under $\prec_C$, but not individual coordinate ranks. It simplifies the probability estimates because the standard Gaussian density depends only on distance from the origin.

To use the stopping arguments, we need to know how often an observation has a small domination risk. The coefficient $a_\Sigma$ measures the leading supply of these points; it also determines the expected frontier size, as described in the introduction. The estimates in Appendix~\ref{app:stopping-estimates}, used with $\Sigma=I_d$ in Sections~\ref{sec:fixed-cutoffs} and~\ref{sec:optimal-stopping}, give
\begin{equation}\label{gauss:gaussian-tail}
H(v):=\Prb\left\{q(Z)\le v\right\}\sim a_\Sigma v(\log(1/v))^{d-1},
\qquad v\downarrow0,
\end{equation}
with the following integral expression for the coefficient:
\begin{equation}\label{gauss:sigma-orthant-coefficient}
a_\Sigma=\frac{\det\Sigma}{(d-1)!}
\int_{(0,\infty)^d}\left(\prod_{k=1}^d z^{(k)}\right)
 e^{-z^{\mathsf T}\Sigma z/2}\,dz.
\end{equation}
For $\Sigma=I_d$, this gives $a_I=1/(d-1)!$, the coefficient in \eqref{eq:product-risk-cdf}.

To estimate how often $q_C(X)$ is small, we first ask which points can only be dominated by points farther from the origin, where Gaussian density is lower. A point $x$ has this property exactly when $x\cdot y\ge0$ for every $y\in C$, since $\|x+y\|^2-\|x\|^2=2x\cdot y+\|y\|^2$. These points form the \emph{dual cone}
\[
C^*=\{x:x\cdot y\ge0\text{ for every }y\in C\}.
\]
Along any fixed direction outside $C^*$, the region $x+C$ contains points closer to the origin by a fixed fraction of $\|x\|$, together with a neighborhood of such points. Their Gaussian density is exponentially higher than at $x$ as $\|x\|$ grows. Thus making $q_C(x)\le v$ requires $x$ to lie farther from the origin, where observations are correspondingly rarer. The proof of Theorem~\ref{gauss:tail-asymptotics} makes this comparison precise and controls directions near the boundary. It shows that these outside directions contribute only a smaller-order term to $\Prb\left\{q_C(X)\le v\right\}$ as $v\downarrow0$. The leading term comes from directions inside $C^*$; integrating their contributions gives the coefficient $a_\Sigma$.

In two dimensions, write $\rho=\Sigma_{12}$ for the correlation.
Then $C$ has opening $\pi-\arccos\rho$ and $C^*$ has opening $\arccos\rho$. Increasing correlation widens the region of dominating points and narrows the directions contributing to the leading term. Figure~\ref{gauss:fig-cones} shows both cones.

\begin{figure}[htb]
\centering
\begin{tikzpicture}[font=\small,>=Stealth]

\foreach \shift/\halfcone/\halfdual/\corr in {
  0/30/60/{-1/2},
  5.8/60/30/{1/2}
} {
  \begin{scope}[xshift=\shift cm]

    \node[font=\normalsize] at (.8,2.75)
      {$\rho=\corr$};

    \draw[black!35]
      (0,0) circle[radius=1.2];

    \node[
      black!60,
      align=center,
      font=\footnotesize,
      inner sep=0pt
    ] at (-.52,.38)
      {higher\\density};

    \begin{scope}[xshift=1.2cm]

      \fill[black!12]
        (0,0) -- (-\halfcone:2.2)
        arc[
          start angle=-\halfcone,
          end angle=\halfcone,
          radius=2.2
        ] -- cycle;

      \draw[black!65,thick]
        (-\halfcone:2.2) -- (0,0)
        -- (\halfcone:2.2);

      \node[
        align=center,
        font=\footnotesize,
        inner sep=0pt
      ] at (1.35,0)
        {$x+C$\\lower\\density};

    \end{scope}

    \draw[black!65,thick,->]
      (0,0) -- (\halfcone:.8);

    \draw[black!65,thick,->]
      (0,0) -- (-\halfcone:.8);

    \node[
      black!65,
      fill=white,
      inner sep=1pt
    ] at (.48,-.12)
      {$C$};

    \draw[ParetoOrange,dashed,thick]
      (-\halfdual:2.35) -- (0,0)
      -- (\halfdual:2.35);

    \node[
      ParetoOrange,
      anchor=south west,
      inner sep=2pt
    ] at (\halfdual:2.35)
      {$C^*$};

    \draw[->,thick]
      (0,0) -- (1.2,0);

    \fill (0,0) circle[radius=1.4pt]
      node[below left] {$0$};

    \fill (1.2,0) circle[radius=1.7pt]
      node[below left] {$x$};

  \end{scope}
}

\end{tikzpicture}
\caption{Gray arrows bound $C$, and orange dashed rays bound $C^*$. For $x$ inside $C^*$, all dominating points in the shaded region lie outside the circle through $x$ and have lower Gaussian density. These directions give the leading term of $\Prb\left\{q_C(X)\le v\right\}$ as $v\downarrow0$.}
\label{gauss:fig-cones}\label{gauss:geometry}
\end{figure}
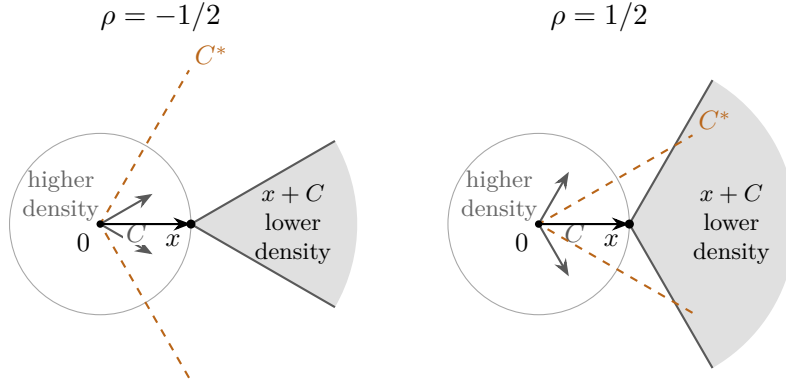

In dimension two, the coefficient has the simpler form
\begin{equation}\label{gauss:rho-tail-coefficient}
a_\rho=1-\frac{\rho\arccos\rho}{\sqrt{1-\rho^2}},\qquad -1<\rho<1.
\end{equation}
It equals one at independence. Its derivation is in the proof of Corollary~\ref{gauss:bivariate-coefficient}; Section~\ref{gauss:comparison-independent} relates the same coefficient to the expected number of maximal observations.

The distribution in \eqref{gauss:gaussian-tail} also fixes the scale of both stopping arguments. Write
\[
M_n=\E\left[|\Max(Z_1,\ldots,Z_n)|\right].
\]
Conditioning on one observation and using the Gaussian mean estimate in Lemma~\ref{gauss:threshold-probability-estimates} gives
\begin{equation}\label{gauss:frontier-size}
M_n=n\E\left[(1-q(Z))^{n-1}\right]\sim a_\Sigma(\log n)^{d-1}.
\end{equation}

In two dimensions, the coefficient $a_\rho$ in \eqref{gauss:rho-tail-coefficient} decreases as correlation increases. Figure~\ref{gauss:fig-frontier-correlation} shows that fewer maximal points correspond to a larger optimal failure probability.

\begin{figure}[H]
\centering
\begin{tikzpicture}
\begin{axis}[name=leftplot,width=.46\textwidth,height=4.8cm,axis lines=left,
 xmin=-.92,xmax=.94,ymin=0,ymax=7.1,
 xlabel={Correlation $\rho$},ylabel={$a_\rho$},
 title={(a) Frontier coefficient},title style={font=\small},
 xtick={-.9,-.5,0,.5,.9},ytick={0,1,2,4,6},
 tick label style={font=\footnotesize},label style={font=\small},
 grid=major,grid style={ParetoGray!15}]
\addplot[very thick,ParetoBlue,domain=-.9:.9,samples=120] {1-x*acos(x)*pi/180/sqrt(1-x*x)};
\addplot[only marks,mark=*,mark size=2pt,ParetoOrange] coordinates {(0,1)};
\end{axis}
\begin{axis}[at={(leftplot.east)},anchor=west,xshift=1.05cm,width=.46\textwidth,height=4.8cm,axis lines=left,
 xmin=-.92,xmax=.94,ymin=0,ymax=32,
 xlabel={Correlation $\rho$},ylabel={$2/a_\rho$},
 title={(b) Sharp failure constant},title style={font=\small},
 xtick={-.9,-.5,0,.5,.9},ytick={0,2,10,20,30},
 tick label style={font=\footnotesize},label style={font=\small},
 grid=major,grid style={ParetoGray!15}]
\addplot[very thick,ParetoBlue,domain=-.9:.9,samples=120] {2/(1-x*acos(x)*pi/180/sqrt(1-x*x))};
\addplot[only marks,mark=*,mark size=2pt,ParetoOrange] coordinates {(0,2)};
\end{axis}
\end{tikzpicture}
\caption{For two Gaussian coordinates with fixed correlation $\rho$, the expected number of maximal points is asymptotic to $a_\rho\log n$, and the optimal failure probability to $2/(a_\rho\log n\log\log n)$. Orange points mark independence.}
\label{gauss:fig-frontier-correlation}
\end{figure}
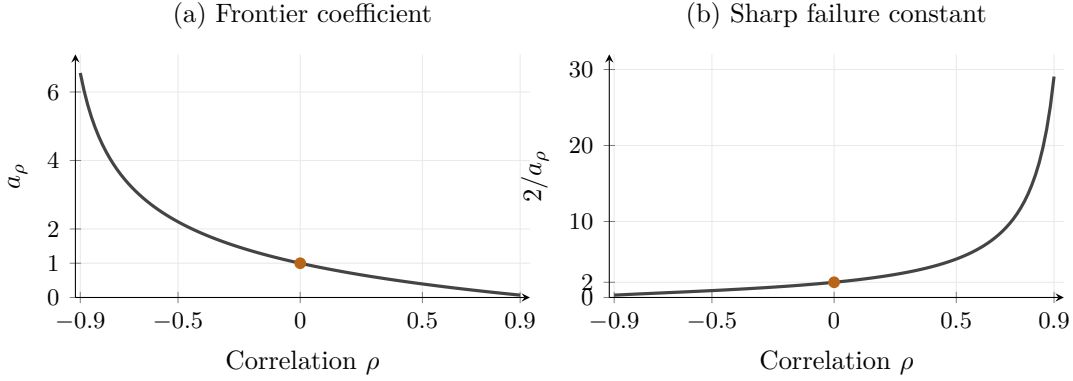

For $u$ in the interval \eqref{gauss:smooth-threshold-band}, \eqref{gauss:gaussian-tail} and \eqref{gauss:patch-sharp-moments} give
\begin{equation}\label{eq:frontier-threshold-estimates}
\Prb\left\{q(Z)\le u\right\}\sim M_nu,\qquad
\E\left[q(Z)\ind\{q(Z)\le u\}\right]\sim\frac{M_nu^2}{2}.
\end{equation}
The first quantity counts points that meet a threshold; the second measures their chance of being dominated later. The proof for fixed-cutoff rules also needs estimates for the random probability $A_s$ of finding a point outside the region dominated by earlier observations. These are also provided by the same appendix. Below we identify the inputs to each stopping proof and explain how their coefficients change with $\Sigma$.

\subsection{Fixed-Cutoff Rules}
\label{gauss:fixed-cutoffs}

\begin{theorem}\label{gauss:cutoff-consistency}
For the fixed-cutoff rule of Section~\ref{sec:fixed-cutoffs}, as $n$ tends to infinity, we have that $
p_n(s_n)\to1$ if and only if
\[
m_n/s_n\to0\quad\text{and}\quad
m_n(\log s_n)^{d-1}/s_n\to\infty,
\]
where $m_n=n-s_n$. Moreover,
\begin{equation*}
\min_{1\le s<n}(1-p_n(s))
\sim\frac{d-1}{a_\Sigma}\frac{\log\log n}{(\log n)^{d-1}},
\end{equation*}
and every minimizing cutoff $s_n^\star$ satisfies
\begin{equation*}
n-s_n^\star\sim\frac{d-1}{a_\Sigma}
\frac{n\log\log n}{(\log n)^{d-1}}.
\end{equation*}
\end{theorem}

\begin{proof}
Apply the proofs of Theorems~\ref{thm:fixed-d-successful-cutoff-families}
and~\ref{thm:cutoff-rate-lower-bound} with
$a_s=a_\Sigma(\log s)^{d-1}/s$.
The failure identity \eqref{eq:cutoff-failure-identity} uses only
independent arrivals and transitivity of domination.
Lemmas~\ref{gauss:cutoff-volume-estimate} and~\ref{lem:cutoff-common-replacements}
supply the same mean, concentration, and replacement bounds for this
fixed $\Sigma$. In particular, $a_\Sigma$ cancels from the mean ratio
in \eqref{eq:cutoff-general-failure}. Also $q(z)>0$ for every finite
$z$, so the argument excluding bounded cutoffs is unchanged.
This proves the criterion for convergence to one.

For $m/s\to0$ and $ma_s\to\infty$, the same concentration and
truncated-ratio bounds justify \eqref{eq:cutoff-later-equivalent},
with later-domination probability asymptotic to $m/s$.
Choose
\begin{equation*}
m_n=\left\lceil\frac{d-1}{a_\Sigma}
\frac{n\log\log n}{(\log n)^{d-1}}\right\rceil.
\end{equation*}
Then $ma_s=(d-1)\log\log s+o(1)$, so
\eqref{gauss:cutoff-no-selection} makes $\Prb\left\{\tau=\infty\right\}=o(m/s)$.
The upper bound follows. Jensen's inequality still forces
$ma_s\ge(d-1-o(1))\log\log n$ at any minimizing cutoff,
as in \eqref{eq:cutoff-no-choice-bounds}. This gives both the lower
bound and the cutoff location.
\end{proof}

\subsection{Optimal Stopping}
\label{gauss:optimal-stopping}

We use the time-dependent threshold from Section~\ref{subsec:product-time-threshold}, with $a_\Sigma$ in place of $1/(d-1)!$. The following theorem shows that it has the optimal leading failure probability. Let $G$ be the distribution of the rank array under $F$.

\begin{theorem}\label{gauss:headline}\label{gauss:cone-stopping}\label{gauss:unknown-correlation}
For fixed $d\ge2$ and positive definite $\Sigma$,
\[
1-p_n^\star(F)\sim1-p^\star_{n,\mathrm{rank}}(G)
\sim\frac{2}{a_\Sigma(d-1)(\log n)^{d-1}\log\log n}.
\]
The full-information rule uses \eqref{gauss:smooth-threshold} below. The rank-based rule constructed below knows that the observations are Gaussian but is not given their means, variances, or correlations. Both have the displayed leading failure probability.
\end{theorem}

\proofstep{Lower bound.}
Apply Theorem~\ref{thm:general-d-fi-lower-bound}'s proof with the same
$\ell$, $\mathcal B_s$, and $\tau'$, using the present $q$ in
\eqref{eq:sharp-lower-q-range}. Its distribution inputs are
\eqref{gauss:gaussian-tail} and \eqref{gauss:laplace-estimates}, with
$a_\Sigma$ replacing $1/(d-1)!$.
For rules with failure probability at most a fixed multiple of the
claimed bound, the four terms in \eqref{eq:sharp-lower-retained-mass}
are still $o(1)$. The final-$\ell$ term uses the mean estimate, the
lower boundary uses $H$, and the upper boundary uses only
$1-(1-q)^m\ge1-e^{-mq}$. Thus $\Prb\left\{\tau'=\infty\right\}=o(1)$.
On $\mathcal B_s$, $\log(1/q)=\log n+O(\log\log n)$, so
\eqref{eq:sharp-lower-cdf-envelope} holds with $a_\Sigma$ replacing
$1/(d-1)!$. The bound $\Prb\left\{Z\in\mathcal B_{n-\ell}\right\}=o(1)$ also persists,
so the integration, Cauchy--Schwarz, and telescoping in
\eqref{eq:sharp-lower-energy} give
\[
1-p(\mathbf r)\ge
\frac{1-o(1)}{2a_\Sigma(\log n)^{d-1}}
\frac{4-o(1)}{(d-1)\log\log n}.
\]
This gives the claimed lower bound.

\proofstep{Full-information upper bound.}
Use the full-information rule of Section~\ref{subsec:product-time-threshold}, with the same $h=n\exp(\sqrt{\log\log n})/(\log n)^{d-1}$ and threshold
\begin{equation}\label{gauss:smooth-threshold}
u_m=\frac{2}{a_\Sigma(m+1)(\log n)^{d-1}\log(1+(m+1)/h)}.
\end{equation}
Figure~\ref{gauss:fig-thresholds} compares the acceptance boundaries after applying $\Phi$ to each coordinate. This drawing preserves coordinate comparisons; it is different from the linear transformation used above to analyze cone comparisons.

\input{figures/gaussian-thresholds}

Corollary~\ref{lem:threshold-common-inputs}, used with $\Sigma=I_d$ in Section~\ref{sec:optimal-stopping}, now applies with the present $\Sigma$ and $\widehat q=q$. At histories reached without a selection and outside its exceptional histories, it gives
\begin{equation}\label{gauss:conditional-threshold-inputs}
\begin{aligned}
\Prb\left\{\tau=s\mid\mathcal H_{s-1}\right\}
&\ge(1-\varepsilon)a_\Sigma(\log n)^{d-1}u_m,\\
\E\left[q(Z_s)\ind\{\tau=s\}\mid\mathcal H_{s-1}\right]
&\le(1+\varepsilon)\frac{a_\Sigma}{2}(\log n)^{d-1}u_m^2.
\end{aligned}
\end{equation}
These are the two inputs in \eqref{eq:threshold-conditional-inputs},
with $a_\Sigma$ replacing $1/(d-1)!$, for the same interval
$\ell\le m\le\lceil n/2\rceil-1$ and with the same exceptional probability. Since $a_\Sigma>0$ is fixed,
the threshold differs from the independent case by a fixed factor,
so the endpoint estimates in that proof still place it within the
corollary's interval. Substituting the threshold in \eqref{gauss:smooth-threshold}
into the first bound cancels $a_\Sigma$.
Thus the probability of reaching each decision has exactly the bounds
\eqref{eq:threshold-reach-bounds}.
In the second bound, the contribution is multiplied by
$1/[a_\Sigma(d-1)!]$. The final contribution is still bounded by
$S_\ell$, while \eqref{gauss:patch-sharp-moments} bounds the initial
contribution by $O_\Sigma(1/[(\log n)^{d-1}(\log\log n)^2])$.
Thus \eqref{eq:threshold-time-conclusion} gives
\[
\limsup_{n\to\infty}(\log n)^{d-1}\log\log n\,(1-p(\mathbf r))
\le\frac{2(1+\varepsilon)}{(1-2\varepsilon)a_\Sigma(d-1)}.
\]
Let $\varepsilon\downarrow0$.

\proofstep{Rank-based rule.}
\label{gauss:gaussian-rank-transfer}\label{gauss:unknown-gaussian-correlation}
Use the initial sample and the marginal probability estimates from Section~\ref{subsec:product-time-threshold}. The product in \eqref{eq:threshold-rank-score} is replaced by a Gaussian rectangle probability. Its marginal inputs are estimated in exactly the same way; the additional unknown quantity is the correlation matrix.

\proofstep{Estimating correlations.} Use the $\lfloor n/4\rfloor$ disjoint pairs among the initial observations. Let $N_{k\ell}$ count pairs for which the two coordinate comparisons agree, so
\[
N_{k\ell}=\sum_{i=1}^{\lfloor n/4\rfloor}
\ind\{(Z_{2i}^{(k)}-Z_{2i-1}^{(k)})
       (Z_{2i}^{(\ell)}-Z_{2i-1}^{(\ell)})>0\}.
\]
The probability that the comparisons agree is $1-\arccos(\Sigma_{k\ell})/\pi$, as shown in the proof of Lemma~\ref{gauss:conservative-estimates}. Thus set
\begin{equation}\label{gauss:conservative-correlation}
\overline{\Sigma}_{k\ell}
=-\cos\left(\pi\overline p_{\lfloor n/4\rfloor}(N_{k\ell})\right)
\quad(k\ne\ell),\qquad \overline{\Sigma}_{kk}=1.
\end{equation}
Only comparisons within each pair are needed.

\proofstep{Estimating marginal positions.} When $\overline\Sigma$ is positive definite, for a later observation $z$ define
\begin{equation}\label{gauss:conservative-rank-score}
\begin{aligned}
\widehat z^{(k)}
&=\Phi^{-1}\left(1-\overline p_{\lfloor n/2\rfloor}
 \bigl(\#\{i\le\lfloor n/2\rfloor:Z_i^{(k)}>z^{(k)}\}\bigr)\right),\\
\widehat q(z)
&=\Prb_{W\sim \mathcal N_d(0,\overline{\Sigma})}
       \left\{W^{(k)}\ge\widehat z^{(k)}\text{ for every }k\in[d]\right\}.
\end{aligned}
\end{equation}
The count estimates the marginal probability of exceeding $z^{(k)}$. Taking its upper confidence bound lowers the standardized coordinate $\widehat z^{(k)}$ and increases the estimated probability of domination. A threshold of $-\infty$ means that this coordinate is unrestricted.

Use the stopping time \eqref{eq:estimated-threshold-rule} with this $\widehat q$ and with $a_{\overline\Sigma}$ replacing $a_\Sigma$ in \eqref{gauss:smooth-threshold}.
\label{gauss:conservative-rank-rule}
If there are no pairs or $\overline\Sigma$ is not positive definite,
wait until time $n$ and select $Z_n$ if it is currently maximal;
otherwise make no selection. Every step uses only comparisons and $n$.

\proofstep{Rank-based upper bound.}
Use Corollary~\ref{lem:threshold-common-inputs} with a smaller error
parameter, say $\varepsilon/4$. Its bounds hold simultaneously for all
thresholds in the stated interval, outside one set of exceptional
histories. Thus they also apply to the threshold chosen from the initial sample.

Lemma~\ref{gauss:conservative-estimates} gives
$a_{\overline\Sigma}/a_\Sigma=1+o(1)$ outside an initial-sample event
of probability $O_d(n^{-1})$. Conditional on that sample, the rule's
threshold is a fixed $(1+o(1))$ multiple of \eqref{gauss:smooth-threshold}.
The endpoint estimates in the proof of Theorem~\ref{thm:threshold-upper-bound}
leave room for this relative change within the permitted interval.
For large $n$, absorbing it into $\varepsilon$ gives
\eqref{gauss:conditional-threshold-inputs} with $u_m$ denoting the
deterministic threshold \eqref{gauss:smooth-threshold}.
The continuation and summation bounds therefore apply unchanged.
The additional exceptional probability is negligible at the claimed
rate, and there are no early selections.

The full-information argument also applies to any fixed closed convex cone $C\subset\mathbb R^d$ with nonempty interior and $C\cap(-C)=\{0\}$. If $H(v)=\Prb\left\{q_C(X)\le v\right\}$, its conclusion is
\begin{equation}\label{gauss:general-cone-bound}
\inf_{\mathbf r}(1-p(\mathbf r))
\sim\frac{2}{(d-1)nH(1/n)\log\log n}.
\end{equation}
Theorem~\ref{gauss:tail-asymptotics} supplies $H$ for these cones,
and Corollary~\ref{lem:threshold-common-inputs} explicitly supplies
the full-information conditional bounds. Since
$nH(1/n)$ is asymptotic to the tail coefficient times $(\log n)^{d-1}$,
using it in the denominator of \eqref{gauss:smooth-threshold} changes
the threshold by $1+o(1)$, handled as above.
This cone formulation also covers the singular models below.

\subsection{Singular Gaussian Distributions}
\label{gauss:singular-gaussian}

We now allow $\Sigma$ to be a singular covariance matrix, including
constant coordinates. Let $r=\operatorname{rank}\Sigma$ and write
$Z=\mu+AX$, where $A\in\mathbb R^{d\times r}$ has column rank $r$,
$AA^{\mathsf T}=\Sigma$, and $X\sim\mathcal N_r(0,I_r)$.
In this subsection, the correlation matrix refers to that of the
nonconstant coordinates. The comparisons are determined by
\[
C=\{z\in\mathbb R^r:Az\ge0\}.
\]
Because $A$ is injective, $C\cap(-C)=\{0\}$. For $r\ge1$, $q_C(x)=q_F(\mu+Ax)$ as above, with standard Gaussian probability now taken in $\mathbb R^r$. All interiors in this subsection are taken in $\mathbb R^r$. For $r\ge2$ and $C$ with nonempty interior, the appendix gives
\[
H(v)=\Prb\left\{q_F(Z)\le v\right\}
\sim a_\Sigma v(\log(1/v))^{r-1},\qquad v\downarrow0,
\]
with $a_\Sigma>0$. We use $a_\Sigma$ for this coefficient also when $\Sigma$ is singular; its integral formula is \eqref{gauss:cone-tail-exact}, with dimension $r$. The proof for fixed-cutoff rules and \eqref{gauss:general-cone-bound} then apply in dimension $r$. For ranks, the only additional issue is estimating $q$ on the support. We resolve it below when the correlation matrix is given.

\begin{proposition}\label{gauss:singular-classification}\label{gauss:known-correlation}
For a fixed positive semidefinite Gaussian covariance $\Sigma$, the following cases exhaust the possibilities.
\begin{enumerate}
\item If $r=0$, all observations coincide and every observation is maximal.
\item If $r\ge1$ and $C$ has empty interior, the observations are pairwise incomparable with probability one, so every observation is maximal.
\item If $r=1$ and $C$ has nonempty interior, the observations form a total order. The rank-based optimum converges to $1/e$, and the full-information optimum converges to $v_\infty=0.580164\ldots$, as defined in Section~\ref{sec:arbitrary-dependence}.
\item If $r\ge2$ and $C$ has nonempty interior, both full information and ranks with known correlation matrix have optimal failure probability asymptotic to
\[
\frac{2}{a_\Sigma(r-1)(\log n)^{r-1}\log\log n}.
\]
\end{enumerate}
\end{proposition}
\begin{proof}
Identical observations do not dominate one another under our convention. A convex cone with empty interior is contained in a proper linear subspace. Since $X_i-X_j\sim \mathcal N_r(0,2I_r)$ has a density, $\Prb\left\{X_i-X_j\in C\cup(-C)\right\}=0$ for each pair, proving the second case. In dimension one, a nontrivial cone containing no line is a half-line, so the problem is exactly one-dimensional best choice. For the final case, delete constant coordinates and standardize the others, writing $\Sigma$ for their correlation matrix.
The rule can identify constant coordinates by equality of their
first two observed values, almost surely. Evaluate the Gaussian rectangle probability in \eqref{gauss:conservative-rank-score} with the known $\Sigma$, which is valid also when $\Sigma$ is singular, and use the thresholds \eqref{gauss:smooth-threshold} with $r$ in place of $d$ in both $h$ and the threshold, after rejecting the first $\lfloor n/2\rfloor$ observations. Corollary~\ref{lem:threshold-common-inputs} applies to this known-matrix estimate in dimension $r$, and gives both bounds in \eqref{gauss:conditional-threshold-inputs} with $r$ in place of $d$. The same summation as in Theorem~\ref{thm:threshold-upper-bound} proves the upper bound, while \eqref{gauss:general-cone-bound} gives the lower bound.
\end{proof}

The geometric condition can be checked directly. After deleting zero rows of $A$, $C$ has nonempty interior if and only if $Az>0$ has a solution. Thus covariance rank alone does not determine the answer. For example, $Z=(X^{(1)},X^{(2)},-X^{(2)})$ with two independent standard normal coordinates has rank two, but $C=\{(x,0):x\ge0\}$ has empty interior. Every observation is maximal with probability one. In two dimensions, correlation $+1$ gives a total order and correlation $-1$ gives pairwise incomparable points. A single nonconstant coordinate also gives a total order.

The fixed-covariance bounds need not hold uniformly near singularity.
For example, let $F_n$ be a bivariate Gaussian distribution with
mean zero, unit marginal variances, and correlation
$\rho_n\in(-1,1)$ satisfying $1-\rho_n=o(n^{-4})$. A pair's two coordinate rankings disagree with probability $\arccos(\rho_n)/\pi$. Thus
\[
\Prb\left\{\exists i<j:
 (Z_i^{(1)}-Z_j^{(1)})(Z_i^{(2)}-Z_j^{(2)})<0\right\}
\le\binom n2\frac{\arccos(\rho_n)}\pi=o(1).
\]
Outside this event, the maximal observation is the first-coordinate maximum. Given the arriving first coordinates, we can generate the second coordinates as $Z_i^{(2)}=\rho_nZ_i^{(1)}+\sqrt{1-\rho_n^2}\,W_i$ with fresh independent standard normal $W_i$. The simulation argument of Section~\ref{sec:arbitrary-dependence} therefore gives $p_n^\star(F_n)\le v_n+o(1)$. For ranks, couple the observations to two identical copies of the first-coordinate ranking; the inputs differ only on the displayed event. Their optimal probabilities therefore differ by $o(1)$. Applying the one-dimensional rules to the first coordinate gives the matching lower bounds, so the limiting optima are $v_\infty$ and $1/e$, respectively.

\subsection{Expected Frontier Size}
\label{gauss:comparison-independent}
We return to a fixed positive definite correlation matrix $\Sigma$.
Equation~\eqref{gauss:frontier-size} gives
$\log M_n\sim(d-1)\log\log n$. Thus Theorems~\ref{gauss:cutoff-consistency} and~\ref{gauss:headline} become
\begin{equation}\label{eq:frontier-cutoff-optimum}
1-\max_{1\le s<n}p_n(s)\sim\frac{\log M_n}{M_n},
\qquad
1-p_n^\star\sim\frac{2}{M_n\log M_n}.
\end{equation}
Here $p_n^\star$ denotes the optimum in either the full-information or the rank-based setting for the distribution in question. The improvement over the best fixed-cutoff rule is consequently asymptotic to $\tfrac12(\log M_n)^2$ for every fixed positive definite $\Sigma$, including $\Sigma=I_d$.

The two estimates in \eqref{eq:frontier-threshold-estimates} explain the shared calculation. At a threshold $u$, an observation meets it with probability about $M_nu$ and then has mean $q$ about $u/2$. With $m$ observations left and $mu\to0$, its chance of later domination is about $mu/2$.

In \eqref{eq:sharp-lower-energy}, the constant $2$ comes from multiplying the factor $(1-o(1))/(2M_n)$ from the distribution of $q$ by the factor $(4-o(1))/\log M_n$ from summing over stopping times. The upper bound uses the matching one-step estimates in \eqref{eq:threshold-conditional-inputs} and sums their contributions over time.

There is one further requirement for that upper bound. A point that meets the threshold may already be dominated by an earlier observation. The bound \eqref{gauss:negligible-covered-mass} controls how many such points are excluded, including when $\Sigma=I_d$. This part uses the geometry of the Gaussian comparison cone. Thus the common distribution estimates explain the shared stopping calculation, but do not by themselves prove the same conclusion for arbitrary dependent coordinates.

\section{Discussion}
\label{sec:discussion}

In our problem, $n$ independent observations from the same distribution
arrive one at a time. Each must be accepted or rejected immediately,
with no recall, and at most one may be selected. Success means selecting
a Pareto-maximal point in the complete sample: no other observation is
at least as good in every coordinate and strictly better in one. Such
a choice can represent a tradeoff among the criteria without being best
in any single coordinate.

The expected number of acceptable observations makes the benefit of
this flexibility precise. Let $M_n$ be the expected size of the Pareto
frontier and $p_n^\star$ the optimal success probability. For every
fixed dimension $d\ge2$, independent continuous coordinates and fixed
nonsingular Gaussian distributions satisfy
\[
1-p_n^\star\sim\frac{2}{M_n\log M_n}.
\]
The relation holds both with full information, where the distribution
is known and numerical values are observed, and with coordinate ranks
alone. Within these families, the growing expected frontier determines
how quickly the failure probability tends to zero. Gaussian correlation
changes the frontier size and, through it, the leading failure constant.

A simple cutoff rule already exploits these opportunities. After
rejecting an initial sample, it selects the first currently maximal
observation, meaning one that no earlier point dominates. Waiting longer
leaves fewer future observations that might dominate the selected point,
but also fewer opportunities to select. We characterize all cutoff
sequences whose success probability tends to one and determine the best
leading failure probability among them.

The time-dependent rule improves on this performance by considering
each candidate's domination risk, the probability that a new independent
observation dominates it. A currently maximal point is accepted only
when this risk falls below a threshold that rises as time runs out,
balancing later domination against the chance of selecting nothing.
The scarcity of observations with very small domination risk yields
a matching lower bound on failure. To attain this bound, the proof shows
that enough points meeting the threshold remain available after
excluding those dominated by earlier observations.

These risk estimates can be obtained from coordinate comparisons.
Under independence, counts from an initial sample estimate the
probability that a fresh observation exceeds a candidate in each
coordinate; their product estimates its domination risk. For nonsingular
Gaussian observations, the rule also estimates the correlations.
These estimates preserve the optimal leading failure probability
without requiring prior knowledge of the marginal distributions
or the Gaussian parameters.

Other forms of dependence can remove this abundance of acceptable
choices. Our counterexample has a joint density and uniform marginals,
with observations lying in cubes arranged so that every point in a higher
cube dominates every point in a lower one. With probability tending
to one, the highest occupied cube contains a single observation, which
is both the unique Pareto maximum and the first-coordinate maximum.
Given each first coordinate, the remaining coordinates can be generated
using fresh independent randomness, so they provide no additional
information for identifying this maximum. Even with full information,
the optimal success probability converges to the one-dimensional limit,
approximately $0.58$.

The two-dimensional Gaussian endpoints show how widely dependence can
change the difficulty of selection. Perfect positive correlation gives
the same ranking in both coordinates and recovers the one-dimensional
problem. Perfect negative correlation makes every observation maximal.

Pareto maximality treats different undominated tradeoffs as equally
successful, although specific preferences may distinguish sharply
between them. If these preferences become known only after the
observations have arrived, a decision maker may want to retain several
candidates as they arrive. For a specified class of preferences,
the question is how closely the best retained candidate can approximate
the best point in the complete sample under the eventual preferences,
and how the guarantee improves with the number of selections. Another
objective keeps a single selection but distinguishes degrees of failure:
require a choice by the end and minimize the expected number of
observations that dominate it.

Dependence across time raises a different issue. Our analysis allows
dependence between coordinates while successive observations are
independent. With temporal dependence, past observations can change
the conditional distribution of future candidates, so domination risk
must be assessed using both the current point and the history already
seen. Understanding how to use this information, and how it changes
the best achievable success probability, provides a further direction
for research.

\section*{Formalization}
A companion formalization in Lean~4 with mathlib is available in the
\href{https://github.com/GO-EPFL/pareto-secretary-formalization}{GitHub repository}.
The repository contains the formal proofs, instructions for checking them,
and a guide connecting the results in this paper to their Lean declarations.

\Needspace{15\baselineskip}
\section*{Declaration of AI usage}
We used OpenAI's ChatGPT and Codex to assist with mathematical proof
development and checking, including the refinement of estimates and
stopping rules. This paper builds on a preliminary version developed by
the authors. In preparing the current version, we determined the extensions
to pursue, guided their mathematical development, and critically evaluated
and revised the AI-assisted arguments.

These tools also assisted with improving the exposition and notation,
checking references, preparing LaTeX source and figures, and developing
and debugging the accompanying Lean~4 formalization. The authors reviewed
the resulting material and take full responsibility for the mathematical
results, proofs, citations, and final manuscript.

\bibliographystyle{alpha}

\phantomsection
\addcontentsline{toc}{section}{References}
\bibliography{references}

\appendix
\crefalias{section}{appendix}

\section{No Numerical Gain with a Known Rank Distribution}\label{app:rank-sufficiency}

When the rank distribution is known but the sampling distribution
is not, numerical values cannot improve the best worst-case guarantee.
We prove Theorem~\ref{thm:rank-sufficiency} and its extension to
arbitrary Borel probability distributions, including distributions
with atoms. In this extension, $\mathcal F(G)$ includes all such
distributions with rank-array distribution $G$.
Equal coordinate values receive equal ranks, so the rank array
records both ties and arrival order.
The reduction uses finite Ramsey theory, as in Correa et al.~\cite{correa2021},
and preserves the given dependence by changing only the numerical
coordinate scales. We use the following finite Ramsey theorem~\cite{ramsey1930}.

\begin{theorem}\label{thm:finite-ramsey}
For positive integers $t,M,K$, there is an integer $R_t(M;K)$ such that every $K$-coloring of the $t$-element subsets of a set of size at least $R_t(M;K)$ is constant on the $t$-element subsets of some $M$-element subset.
\end{theorem}
For a finite set $X$, write
$[X]^t=\{A\subseteq X:|A|=t\}$ for its $t$-element subsets.
\begin{lemma}\label{lem:product-ramsey}
Fix positive integers $t_1,\ldots,t_d,M,K$, with $M\ge\max_k t_k$. There is a finite size such that, for ordered coordinate sets $X_k$ at least that large, every coloring $\prod_k[X_k]^{t_k}\to[K]$ is constant on $\prod_k[Y_k]^{t_k}$ for some $M$-element subsets $Y_k\subseteq X_k$.
\end{lemma}
\begin{proof}
Induct on $d$, the case $d=1$ being the finite Ramsey theorem. Choose $Q$ large enough for the induction hypothesis on coordinates $2,\ldots,d$, and restrict those coordinate sets to size $Q$. Color a $t_1$-subset of $X_1$ by the profile of all its original colors over the remaining coordinate subsets. There are at most $K^{\prod_{k=2}^d\binom Q{t_k}}$ profiles. The finite Ramsey theorem gives an $M$-element subset of $X_1$ on which that profile is constant. Apply the induction hypothesis to its common profile to select the other $M$-element subsets.
\end{proof}

\begin{corollary}\label{cor:simultaneous-product-ramsey}
For finitely many such colorings, possibly with different subset sizes in each coordinate, sufficiently large coordinate sets contain $M$-element subsets on which all the colorings are constant.
\end{corollary}
\begin{proof}
Choose the required sizes backward through the finite list of colorings, and apply Lemma~\ref{lem:product-ramsey} successively. Restriction preserves the constancy of the earlier colorings.
\end{proof}

\begin{proof}[Proof of Theorem~\ref{thm:rank-sufficiency}]
Fix $F_0$, a numerical rule $\mathbf r$, and $\varepsilon>0$. All transformations below are chosen before the observations are drawn.

\proofstep{Cell partition.}
For any $h>0$, we can partition each coordinate line into finitely many ordered intervals and individual points, with probability at most $h$ in each interval. Take marginal quantiles at a mesh of size at most $h$, make each distinct endpoint a separate cell, and use the intervening open intervals as the other cells. The two unbounded end intervals also have mass at most $h$. There are finitely many cells. Let $M$ be at least the number of cells in every coordinate and at least $n$.

For an i.i.d. sample $Z_1,\ldots,Z_n\sim F_0$, let $D$ be the event that distinct observed values in each coordinate lie in different cells. Equal values are allowed and remain equal. For a pair of observations, the probability of distinct values in the same interval cell is at most the sum of the squared cell probabilities, at most $h$. Hence
\begin{equation*}
 \Prb\left\{D^c\right\}\le d\binom n2 h.
\end{equation*}
This uses independence across observations, not across coordinates.

\proofstep{Ramsey reduction.}
For a time-$t$ rank array $g$, including its ties and arrival order, let $m_k(g)$ be the number of distinct values in coordinate $k$. The ordered lists of distinct coordinate values, together with $g$, determine all the numerical observations through time $t$. Thus $r_t$ becomes a bounded measurable function $f_{t,g}$ of $\sum_k m_k(g)$ real variables on the open domain where each coordinate list is strictly increasing.

Fix $\delta>0$. Choose sufficiently large finite ordered grids $X_k\subset(0,1)$, and color each tuple in $\prod_k[X_k]^{m_k(g)}$ by a subinterval of $[0,1]$ of length at most $\delta$ containing $f_{t,g}$ at that tuple. Do this for every $t\le n$ and every rank array $g$. By the Lebesgue differentiation theorem, almost every tuple is a
Lebesgue point of the corresponding bounded measurable function.
Choose the grids by independent sampling from continuous densities.
Every tuple then has a density, so, with probability one, all the
finitely many tuples used are Lebesgue points.

Corollary~\ref{cor:simultaneous-product-ramsey} gives $M$-element subgrids $S_k$ on which all these colorings are constant. Choose a number $r'_t(g)$ in the common color interval. Then
\begin{equation}\label{eq:rank-general-grid}
 |f_{t,g}(x)-r'_t(g)|\le\delta
\end{equation}
at every grid tuple with rank array $g$. The numbers $r'_t(g)$ define a rank-based stopping rule. Assign the ordered cells of coordinate $k$ to distinct, ordered points of $S_k$; unused grid points do not matter.

\proofstep{Monotone transformations.}
In an interval cell assigned to grid point $a$, take
$v(x)=2\arctan(x)/\pi$; for a cell containing one point take $v=0$.
Define the coordinate transformation on that cell by
\[
 T(x)=a+U+bv(x),\qquad U\sim\operatorname{Unif}(-b,b).
\]
Use independent offsets for every coordinate and cell. Choose $b>0$ so small that the intervals $(a-2b,a+2b)$ are disjoint and ordered. For every fixed choice of the offsets, each $T_k$ is strictly increasing on all of $\mathbb R$. It preserves comparisons and ties exactly, so the transformed i.i.d. distribution $F_T$ belongs to $\mathcal F(G)$. On each product of interval cells, the transformation has a
continuously differentiable inverse. Thus, whenever $F_0$ is
absolutely continuous, so is $F_T$ for every fixed choice of the
offsets; the cell boundaries have probability zero.
The construction therefore also stays within the absolutely
continuous class used in Theorem~\ref{thm:rank-sufficiency}.

Fix an original sample in $D$. Distinct coordinate values of any prefix use different cell offsets. Thus, the transformed lists of distinct values have a joint density, with independent coordinates in this list, bounded by $(2b)^{-\sum_k m_k(g)}$ on the cube of radius $2b$ around their grid tuple. This independence concerns the random offsets only; no independence of the original sample's coordinates is assumed. Repeated values use the same offset and are represented by a single variable in $f_{t,g}$.

There are finitely many grid tuples, all Lebesgue points. We can therefore decrease $b$ so that the average of $|f_{t,g}(y)-f_{t,g}(a)|$ over every such cube is at most $2^{-dn}\delta$. The density bound, whose ratio to the uniform cube density is at most $2^{dn}$, then gives, uniformly over original samples in $D$,
\[
 \E_T\left[|r_t(T(Z_1),\ldots,T(Z_t))-f_{t,g}(a)|\right]\le\delta.
\]
Together with \eqref{eq:rank-general-grid}, this bounds the expected difference from $r'_t(g)$ by $2\delta$.

Run both stopping rules using the same extra uniform random variables. The probability of a first disagreement is at most the sum of the differences of their stopping probabilities. Therefore
\[
 \E_T\left[|p_{F_T}(\mathbf r)-p_{F_T}(\mathbf r')|\right]
 \le d\binom n2 h+2n\delta.
\]
Choose $h,\delta>0$ so that this is at most $\varepsilon/2$. Every $F_T$ has the same rank distribution, hence the same value for $\mathbf r'$. Some fixed choice of the transformations satisfies the stated inequality with error $\varepsilon$. The transformed observations are i.i.d. for each such fixed choice; the mixture over transformations is used only to prove existence.

Thus every numerical rule has worst-case value at most $p^\star_{n,\mathrm{rank}}(G)+\varepsilon$. Letting $\varepsilon\downarrow0$ gives one minimax inequality. A rank-based optimal rule has the same value under every $F\in\mathcal F(G)$, giving the reverse inequality. Such an optimal rule exists because there are finitely many rank histories and decisions. The same argument applies to the smaller family consisting only of increasing coordinate transformations of any fixed $F_0$, which proves the stronger assertion of the theorem.
\end{proof}

\section{Gaussian Probability Estimates}
\label{app:stopping-estimates}

The estimates below apply to Sections~\ref{sec:fixed-cutoffs}, \ref{sec:optimal-stopping}, and \ref{sec:gaussian}. Independent coordinates correspond to $\Sigma=I_d$, or $C=\mathbb R_+^d$, with $a_I=1/(d-1)!$. We first estimate the distribution of $q$ and the probability left undominated, then bound the errors from estimating probabilities using coordinate ranks.

\Needspace{4\baselineskip}
\bigskip
\noindent\textbf{Moment identities}

\nopagebreak

We first record identities valid for any distribution of the observations. For any observation $Z$, let $H(v)=\Prb\left\{q(Z)\le v\right\}$. Integration by parts gives, for $0<u\le1$, $s\ge2$, and $\lambda>0$,
\begin{equation*}
\begin{aligned}
\E\left[q(Z)\ind\{q(Z)\le u\}\right]
 &=uH(u)-\int_0^u H(v)\,dv,\\
\E\left[(1-q(Z))^{s-1}\right]
 &=(s-1)\int_0^1(1-v)^{s-2}H(v)\,dv,\\
\E\left[e^{-\lambda q(Z)}\right]
 &=e^{-\lambda}+\lambda\int_0^1e^{-\lambda v}H(v)\,dv.
\end{aligned}
\end{equation*}
If $H(v)/(v(\log(1/v))^{d-1})$ tends to a positive finite constant as $v\downarrow0$, these identities imply
\begin{equation}\label{eq:distribution-moment-asymptotics}
\begin{aligned}
\E\left[q(Z)\ind\{q(Z)\le u\}\right]&\sim\tfrac12uH(u),\\
\E\left[(1-q(Z))^{s-1}\right]&\sim H(1/s),\\
\E\left[e^{-\lambda q(Z)}\right]&\sim H(1/\lambda),\\
\E\left[q(Z)e^{-\lambda q(Z)}\right]&=O\bigl(H(1/\lambda)/\lambda\bigr).
\end{aligned}
\end{equation}
Here $u\downarrow0$, $s\to\infty$, and $\lambda\to\infty$, respectively. The first three estimates follow by dominated convergence after the substitutions $v=ut$, $v=t/s$, and $v=t/\lambda$. Integrable bounds are constant multiples of $t(1+|\log t|)^{d-1}$ on $(0,1)$ and $te^{-t/2}(1+|\log t|)^{d-1}$ on $(0,\infty)$. The last uses $xe^{-\lambda x}\le2(e\lambda)^{-1}e^{-\lambda x/2}$.

The same substitution gives the moments used to analyze fixed-cutoff rules. For $k=0,1$ and fixed $T\in(0,\infty]$,
\begin{equation}\label{eq:distribution-cutoff-moments}
\int_0^{\min\{T/s,1\}}v^k(1-v)^s\,dH(v)
\sim\frac{H(1/s)}{s^k}\int_0^T t^ke^{-t}\,dt.
\end{equation}
Here $\min\{T/s,1\}=1$ when $T=\infty$. Integration by parts and the same integrable bounds justify the limit; the boundary term at $T/s$ is included when $T<\infty$.

The following bound also applies to other distributions and transitive comparisons. Write $A_s=F(\mathcal U_s)$ and assume $A_s>0$ almost surely. If $z\in \mathcal U_s$, every point dominating $z$ also belongs to $\mathcal U_s$, so $q(z)\le A_s$. Consequently,
\begin{equation}\label{eq:undominated-denominator-bound}
\begin{aligned}
0\le A_s^{-2}\int_{\mathcal U_s}q\,dF&\le1,\\
\E\left[A_s^{-2}\int_{\mathcal U_s}q\,dF\right]
&\le\frac4{(\E\left[A_s\right])^2}\E\left[\int_{\mathcal U_s}q\,dF\right]
   +\Prb\left\{A_s<\E\left[A_s\right]/2\right\}.
\end{aligned}
\end{equation}
The second line splits according to $A_s\ge\E\left[A_s\right]/2$.

\Needspace{4\baselineskip}
\bigskip
\noindent\textbf{Gaussian tail asymptotics}

\nopagebreak

In this appendix, $\|x\|$ is the Euclidean norm and matrix norms are its induced operator norms. We write $\operatorname{int}(C)$ for the interior of $C$, $B(0,t)=\{x:\|x\|<t\}$, and $\vol$ for Lebesgue volume. Let $\gamma$ be standard Gaussian probability on $\mathbb R^d$, with density $\phi(x)=(2\pi)^{-d/2}e^{-\|x\|^2/2}$. Fix a closed convex cone $C$ with nonempty interior and $C\cap(-C)=\{0\}$. As in Section~\ref{sec:gaussian},
\[
q_C(x)=\gamma(x+C),\qquad H(v)=\Prb\left\{q_C(X)\le v\right\},\qquad X\sim\gamma.
\]
We use the dual cone $C^*$ defined in Section~\ref{gauss:cone-stopping-section}. For $x\in\operatorname{int}(C^*)$, define
\[
J(x)=\int_Ce^{-x\cdot z}\,dz.
\]
It is finite since $x\cdot z$ is bounded below by a positive multiple of $\|z\|$ on $C$. Scaling gives $J(tx)=t^{-d}J(x)$.

\begin{theorem}\label{gauss:tail-asymptotics}
For a fixed cone $C$ as above, in any dimension $d\ge1$,
\begin{equation}\label{gauss:cone-tail-exact}
H(v)\sim v(\log(1/v))^{d-1}
 \left(2^{d-1}\int_{S^{d-1}\cap\operatorname{int}(C^*)}
                 \frac{d\sigma(\theta)}{J(\theta)}\right).
\end{equation}
Here $S^{d-1}$ is the unit sphere and $\sigma$ is its surface measure; for $d=1$, it is counting measure on $\{-1,1\}$. The integral is finite and strictly positive.
\end{theorem}

\begin{proof}
We compute the radial contribution in each direction, then use a common bound to integrate over directions. Translating the Gaussian density gives
\begin{equation}\label{gauss:cone-tail-identity}
 q_C(r\theta)=\phi(0)e^{-r^2/2}
       \int_C\exp\{-r\theta\cdot z-\|z\|^2/2\}\,dz.
\end{equation}
Substituting $y=rz$ and applying dominated convergence gives, uniformly on each compact subset of $\operatorname{int}(C^*)\cap S^{d-1}$,
\begin{equation}\label{gauss:cone-corner-laplace}
 q_C(r\theta)\sim \phi(0)J(\theta)e^{-r^2/2}r^{-d}.
\end{equation}
Here domination follows from $\theta\cdot y\ge c\|y\|$ on each such compact set. Write the radial contribution to $H(v)$ as
\[
 A_v(\theta)=\phi(0)\int_0^\infty
  \ind\{q_C(r\theta)\le v\}e^{-r^2/2}r^{d-1}\,dr.
\]
If $t_v$ solves $\phi(0)J(\theta)e^{-t_v^2/2}t_v^{-d}=v$, then $t_v^2\sim2\log(1/v)$ and the Gaussian radial tail satisfies
\[
 \phi(0)\int_{t_v}^\infty e^{-r^2/2}r^{d-1}\,dr
 \sim \phi(0)e^{-t_v^2/2}t_v^{d-2}
 =\frac{v\,t_v^{2d-2}}{J(\theta)}.
\]
Sandwiching \eqref{gauss:cone-corner-laplace} between its $(1\pm\varepsilon)$ multiples gives
\begin{equation}\label{gauss:cone-angular-limit}
 \frac{A_v(\theta)}{v(\log(1/v))^{d-1}}
 \longrightarrow\frac{2^{d-1}}{J(\theta)}
 \quad\text{for }\theta\in\operatorname{int}(C^*).
\end{equation}
For these upper and lower bounds, any bounded interval of radii contributes nothing once $v$ is small enough, because $q_C$ is continuous and strictly positive.

Next fix $\theta\notin C^*$. Choose $z_0\in\operatorname{int}(C)$ with $\theta\cdot z_0<0$, and then a small $\lambda>0$ such that $\|\theta+\lambda z_0\|<1$. For large $r$, the ball $r\lambda z_0+B(0,1)$ lies in $C$. Choose $\beta\in(\|\theta+\lambda z_0\|^2,1)$. Integrating the Gaussian density over the corresponding dominating ball gives, for large $r$,
\[
q_C(r\theta)\ge
\int_{r(\theta+\lambda z_0)+B(0,1)}\phi(0)e^{-\|z\|^2/2}\,dz
\ge c_{C,\theta} e^{-\beta r^2/2}
\]
for some $c_{C,\theta}>0$. Thus $q_C(r\theta)\le v$ requires $r^2\ge(2/\beta)\log(c_{C,\theta}/v)$ once $v$ is small. The resulting radial contribution is $O_{C,\theta}(v^{1/\beta}(\log(1/v))^{(d-2)/2}) =o(v(\log(1/v))^{d-1})$.

To integrate these limits across all directions, restrict the integral in \eqref{gauss:cone-tail-identity} to $(1+r)^{-1}(C\cap B(0,1))$. For every direction and every $r\ge0$,
\begin{equation}\label{gauss:cone-uniform-lower}
q_C(r\theta)\ge c_C\frac{e^{-r^2/2}}{(1+r)^d},\qquad
c_C=\phi(0)e^{-3/2}\vol(C\cap B(0,1))>0.
\end{equation}
With this fixed $c_C$, let $r_v$ solve $c_Ce^{-r_v^2/2}(1+r_v)^{-d}=v$. The right side of \eqref{gauss:cone-uniform-lower} is decreasing in $r$; therefore $q_C(r\theta)\le v$ implies $r\ge r_v$, for every $\theta$. Moreover,
\[
 r_v^2=2\log(1/v)-d\log\log(1/v)+O_C(1).
\]
Thus, uniformly over the sphere,
\[
 A_v(\theta)\le \phi(0)\int_{r_v}^\infty e^{-r^2/2}r^{d-1}\,dr
 \le O_C\bigl(v(1+r_v)^dr_v^{d-2}\bigr)
 =O_C\bigl(v(\log(1/v))^{d-1}\bigr).
\]
This also holds for $d=1$, using the usual one-dimensional Mills bound. The boundary of the full-dimensional dual cone has zero sphere measure. Dominated convergence applied to $H(v)=\int_{S^{d-1}}A_v(\theta)\,d\sigma(\theta)$ proves \eqref{gauss:cone-tail-exact}.  The coefficient is finite by this same uniform bound and strictly positive because the dual cone has nonempty interior.
\end{proof}

For a Gaussian distribution, write $Z=AX$ and $\Sigma=AA^{\mathsf T}$, where $A$ has $d$ linearly independent columns. If $C=\{x:Ax\ge0\}$ has nonempty interior, the coefficient in \eqref{gauss:cone-tail-exact} is $a_\Sigma$, since $q_C(X)=q_F(Z)$. Thus $d$ here is the rank of $\Sigma$ and can be smaller than the number of coordinates of $Z$. We use this dimension convention through Lemma~\ref{lem:cutoff-common-replacements}. For a general cone, the same proofs use the coefficient displayed in \eqref{gauss:cone-tail-exact} instead.

For a compact $P\subset\operatorname{int}(C^*)\cap S^{d-1}$, put
\begin{equation*}
a_P=2^{d-1}\int_P\frac{d\sigma(\theta)}{J(\theta)}.
\end{equation*}
Equation~\eqref{gauss:cone-angular-limit} gives the corresponding asymptotic restricted to directions in $P$. Enlarging $P$ within the interior of $C^*$ makes $a_P$ approach $a_\Sigma$.

\begin{corollary}\label{gauss:pd-tail-exact}\label{gauss:bivariate-coefficient}
For $Z\sim \mathcal N_d(\mu,\Sigma)$ with positive definite $\Sigma$, \eqref{gauss:gaussian-tail} holds with the coefficient in \eqref{gauss:sigma-orthant-coefficient}. Equivalently,
\begin{equation}\label{gauss:sigma-angular-coefficient}
a_\Sigma=2^{d-1}(\det\Sigma)^{1/2}
  \int_{\substack{\theta\in S^{d-1}\\\Sigma^{-1/2}\theta>0}}
       \prod_{k=1}^d(\Sigma^{-1/2}\theta)^{(k)}\,d\sigma(\theta).
\end{equation}
Positive diagonal rescaling leaves this coefficient unchanged, and $a_I=1/(d-1)!$. In dimension two, \eqref{gauss:rho-tail-coefficient} gives the coefficient as a function of the correlation.
\end{corollary}

\begin{proof}
After translation, write $Z=\Sigma^{1/2}X$ with $X\sim \mathcal N_d(0,I_d)$. Substitute $z=\Sigma^{-1/2}y$ in the integral defining $J$ to obtain
\[
J(\theta)=(\det\Sigma)^{-1/2}
\prod_{k=1}^d\bigl[(\Sigma^{-1/2}\theta)^{(k)}\bigr]^{-1}
\]
for directions inside the dual cone. This proves \eqref{gauss:sigma-angular-coefficient}. To obtain \eqref{gauss:sigma-orthant-coefficient}, substitute $z=\Sigma^{-1/2}t\theta$ and integrate in $t$ to obtain
\[
\begin{aligned}
&\int_{(0,\infty)^d}\left(\prod_{k=1}^d z^{(k)}\right)
                                  e^{-z^{\mathsf T}\Sigma z/2}\,dz\\
&\quad=(\det\Sigma)^{-1/2}2^{d-1}(d-1)!
\int_{\substack{\theta\in S^{d-1}\\\Sigma^{-1/2}\theta>0}}
                 \prod_{k=1}^d(\Sigma^{-1/2}\theta)^{(k)}\,d\sigma(\theta).
\end{aligned}
\]
The radial integral is $2^{d-1}(d-1)!$. Diagonal rescaling preserves comparisons, and the integral factors into $d$ unit integrals when $\Sigma=I$.

For the bivariate formula, the geometry in Figure~\ref{gauss:fig-cones} can be obtained by writing $Z=(a_1\cdot X,a_2\cdot X)$, where $a_1=(\cos\alpha,\sin\alpha)$, $a_2=(\cos\alpha,-\sin\alpha)$, and $\alpha=\tfrac12\arccos\rho$. The comparison cone $\{x:a_1\cdot x\ge0,\ a_2\cdot x\ge0\}$ has opening $\pi-\arccos\rho$, and its dual has opening $\arccos\rho$. In \eqref{gauss:sigma-orthant-coefficient}, substitute $z^{(2)}=t z^{(1)}$ and integrate over $z^{(1)}$. This gives
\[
a_\rho=2(1-\rho^2)\int_0^\infty
                 \frac{t\,dt}{(t^2+2\rho t+1)^2}.
\]
Since
\[
\int_0^\infty\frac{dt}{t^2+2\rho t+1}
=\frac{\arccos\rho}{\sqrt{1-\rho^2}},
\]
differentiating in $\rho$ proves \eqref{gauss:rho-tail-coefficient}. In particular, $a_0=1$, $a_\rho\sim\tfrac23(1-\rho)$ as $\rho\uparrow1$, and $a_\rho\sim\pi/\sqrt{2(1+\rho)}$ as $\rho\downarrow-1$.
\end{proof}

\Needspace{4\baselineskip}
\bigskip
\noindent\textbf{Domination by earlier observations}

\nopagebreak

The next bound is used both to control the points excluded by the sample and to bound the variance of the probability left undominated. For the explicit cutoff, its constant must remain bounded while the set of directions expands towards the boundary of $C^*$. For all sufficiently large $n$, put
\[
P_n=\left\{\theta\in S^{d-1}:\theta\cdot z\ge
       \frac{\|z\|}{\log\log n}\text{ for every }z\in C\right\}.
\]
Every fixed compact subset of $\operatorname{int}(C^*)\cap S^{d-1}$ is contained in $P_n$ for all sufficiently large $n$.

\begin{lemma}\label{gauss:covered-probability}
Fix $d\ge2$ and a cone $C$ as above. There is a constant $c_C$ such that, for all sufficiently large $n$, all $n^{-1}(\log n)^{-22d}\le u\le n^{-1}(\log n)^2$, and all $\|x_i\|\le\sqrt{2\log n+40d\log\log n}$,
\begin{equation}\label{gauss:covered-weighted}
\begin{aligned}
&\Prb\left\{q_C(X)\le u,\ X/\|X\|\in P_n,\ x_i-X\in C\right\}\\
&\qquad\le c_Cu\,\ind\{q_C(x_i)\le u\}
              \left(1+\log\frac{u}{q_C(x_i)}\right)^d.
\end{aligned}
\end{equation}
The right side is defined to be zero when $q_C(x_i)>u$.
\end{lemma}

\begin{proof}
We first make \eqref{gauss:cone-corner-laplace} quantitative. If $\theta\cdot z\ge\delta\|z\|$ for all $z\in C$, then $1-e^{-\|z\|^2/(2t^2)}\le\|z\|^2/(2t^2)$ gives
\begin{equation}\label{gauss:cone-quantitative-laplace}
\begin{aligned}
0\le1-\frac{q_C(t\theta)}{\phi(t\theta)J(t\theta)}
&\le\frac{1}{2t^2\delta^2J(\theta)}
          \int_C(\theta\cdot z)^2e^{-\theta\cdot z}\,dz\\
&=\frac{d(d+1)}{2t^2\delta^2}.
\end{aligned}
\end{equation}
The equality follows by differentiating $J(a\theta)=a^{-d}J(\theta)$ twice in $a$ at $a=1$.

Now take $\delta=1/\log\log n$. A point counted on the left of \eqref{gauss:covered-weighted} has the form $x=t\theta$, with $\theta\in P_n$ and $y=x_i-x\in C$. Put $r=\|x_i\|$. The bound \eqref{gauss:cone-uniform-lower} and $q_C(x)\le u$ give $t^2\ge2\log n-O_C(\log\log n)$. Also
\begin{equation*}
r^2-t^2=2t\theta\cdot y+\|y\|^2\ge2\delta t\|y\|.
\end{equation*}
Consequently $t\le r$ and
\[
\|y\|=O_C\left(\frac{(\log\log n)^2}{\sqrt{\log n}}\right),
\qquad
\left\|\frac{x_i}{r}-\theta\right\|
=O_C\left(\frac{(\log\log n)^2}{\log n}\right)=o(\delta).
\]
Thus the direction $x_i/r$ satisfies the same inner-product bound with $\delta/2$, so \eqref{gauss:cone-quantitative-laplace} applies to both $x$ and $x_i$. For $z\in C$,
\[
|y\cdot z|\le\frac{2\|y\|}{r\delta}\,x_i\cdot z
=O_C\left(\frac{(\log\log n)^3}{\log n}\right)x_i\cdot z.
\]
Since $x=x_i-y$, bounding the two exponential integrands by these multiples of $x_i\cdot z$ and using the scaling identity for $J$ gives
\[
\frac{J(x)}{J(x_i)}
=1+O_C\left(\frac{(\log\log n)^3}{\log n}\right).
\]
Combining this with \eqref{gauss:cone-quantitative-laplace} yields
\[
\frac{q_C(x)}{q_C(x_i)}
=e^{(r^2-t^2)/2}
 \left[1+O_C\left(\frac{(\log\log n)^3}{\log n}\right)\right].
\]
The inclusion $x_i+C\subseteq x+C$ gives $q_C(x_i)\le q_C(x)\le u$. As $\|y\|^2=o(1)$, the preceding estimates also give
\[
x_i\cdot y=\frac{r^2-t^2+\|y\|^2}{2}
\le\log\frac{u}{q_C(x_i)}+o(1).
\]
All error bounds above hold for every point being counted, with constants independent of that point.

For all sufficiently large $n$, we may therefore integrate over the larger set of $y\in C$ with $x_i\cdot y\le1+\log(u/q_C(x_i))$. The scaling of the cone gives
\[
\vol\{y\in C:x_i\cdot y\le t\}=\frac{J(x_i)t^d}{d!},\qquad t\ge0.
\]
Using $\phi(x_i-y)\le\phi(x_i)e^{x_i\cdot y}$ and integrating over these slices gives
\[
\begin{aligned}
&\Prb\left\{q_C(X)\le u,\ X/\|X\|\in P_n,\ x_i-X\in C\right\}\\
&\quad\le\frac{\phi(x_i)J(x_i)}{(d-1)!}
       \int_0^{1+\log(u/q_C(x_i))}e^t t^{d-1}\,dt\\
&\quad\le c_Cu\left(1+\log\frac{u}{q_C(x_i)}\right)^d.
\end{aligned}
\]
The last inequality uses \eqref{gauss:cone-quantitative-laplace}; if there is no point to count, the bound is immediate.
\end{proof}

\Needspace{4\baselineskip}
\bigskip
\noindent\textbf{Moments and sample bounds}

\nopagebreak

\begin{lemma}\label{gauss:threshold-probability-estimates}
Let $C=\{x:Ax\ge0\}\subset\mathbb R^d$ be associated with $\Sigma=AA^{\mathsf T}$ as above, with $d\ge2$, and let $X\sim \mathcal N_d(0,I_d)$. For a fixed compact $P\subset\operatorname{int}(C^*)\cap S^{d-1}$, for all $n^{-1}(\log n)^{-2(d-1)}\le u\le n^{-1}$,
\begin{equation}\label{gauss:patch-sharp-moments}
\begin{aligned}
\Prb\left\{q_C(X)\le u,\ X/\|X\|\in P\right\}
  &=(a_P+o(1))u(\log n)^{d-1},\\
\E\left[q_C(X)\ind\{q_C(X)\le u\}\right]
  &=(a_\Sigma/2+o(1))u^2(\log n)^{d-1}.
\end{aligned}
\end{equation}
The errors tend to zero independently of $u$. The same estimates hold if either endpoint of the interval is multiplied by a fixed positive constant. Also, as $n\to\infty$,
\begin{equation}\label{gauss:laplace-estimates}
\E\left[(1-q_C(X))^{n-1}\right]\sim\frac{a_\Sigma(\log n)^{d-1}}n.
\end{equation}
The exponential and weighted exponential bounds follow from \eqref{eq:distribution-moment-asymptotics} and \eqref{gauss:F-order}. Among $n$ independent such observations, the probability that either of the following bounds fails is $o(((\log n)^{d-1}\log\log n)^{-1})$, where the bounds are
\begin{equation}\label{gauss:sample-bounds}
\begin{gathered}
\max_{i\le n}\|X_i\|\le\sqrt{2\log n+40d\log\log n},\\
\#\{i\le n:q_C(X_i)\le n^{-1}e^{-\sqrt{\log\log n}/3}\}
\le\left\lceil(\log n)^{d-1}e^{-\sqrt{\log\log n}/4}\right\rceil.
\end{gathered}
\end{equation}
\end{lemma}

\begin{proof}
Theorem~\ref{gauss:tail-asymptotics} and its restriction to $P$ give
$H(u)\sim a_\Sigma u(\log(1/u))^{d-1}$ and the same formula with
$a_P$ for points in $P$. In particular,
\begin{equation}\label{gauss:F-order}
H(u)=O_C\bigl(u(1+\log(1/u))^{d-1}\bigr),\qquad 0<u\le1.
\end{equation}
The moment assertions now follow from
\eqref{eq:distribution-moment-asymptotics} and
$\log(1/u)=\log n+O(\log\log n)$ throughout the stated interval.

\proofstep{Sample bounds.}
The Gaussian radial tail is $\Prb\left\{\|X\|>t\right\}=O_d(t^{d-2}e^{-t^2/2})$ as $t\to\infty$. A union bound at the radius in \eqref{gauss:sample-bounds} gives $O_d((\log n)^{(d-2)/2-20d})$. The count in its second line is binomial, with mean
\[
O_C\left((\log n)^{d-1}e^{-\sqrt{\log\log n}/3}\right).
\]
Applying \eqref{eq:binomial-count-tail} at the upper bound on the count in \eqref{gauss:sample-bounds} gives error $o(((\log n)^{d-1}\log\log n)^{-1})$. A violation can be recorded when the observations already seen first exceed either bound. The stopping proofs therefore do not condition on an event involving future observations.
\end{proof}

\begin{lemma}\label{gauss:sample-exclusion}
Fix $d\ge2$, a cone $C$ as in Lemma~\ref{gauss:covered-probability}, and a compact $P\subset\operatorname{int}(C^*)\cap S^{d-1}$. A history is exceptional if the observations already seen violate either bound in \eqref{gauss:sample-bounds}. The probability of any such history is $o(((\log n)^{d-1}\log\log n)^{-1})$. At every other history, for every $s\le n$ and
\begin{equation}\label{gauss:smooth-threshold-band}
\frac1{n(\log n)^{2(d-1)}}\le u\le
\frac{e^{-\sqrt{\log\log n}/3}}n,
\end{equation}
we have, using $\mathcal H_{s-1}$ from Section~\ref{sec:model}
with $X_i$ in place of $Z_i$,
\begin{equation}\label{gauss:negligible-covered-mass}
\Prb\left\{q_C(X_s)\le u,\ X_s/\|X_s\|\in P,\
       \exists i<s:\ X_s\prec_C X_i\mid\mathcal H_{s-1}\right\}
=o\bigl(u(\log n)^{d-1}\bigr).
\end{equation}
The error is uniform in $s,u$ and these histories.
\end{lemma}
\begin{proof}
The norm bound and \eqref{gauss:cone-uniform-lower} imply $q_C(X_i)\ge n^{-1}(\log n)^{-c_C}$ for a fixed $c_C$. Also $P\subseteq P_n$ for all sufficiently large $n$. Thus every logarithm in \eqref{gauss:covered-weighted} is $O_C(\log\log n)$ for the earlier points that can contribute. Adding the single-point bound over the earlier observations and using the sample count in \eqref{gauss:sample-bounds} gives
\[
\begin{aligned}
&\Prb\left\{q_C(X_s)\le u,\ X_s/\|X_s\|\in P,\
       \exists i<s:\ X_s\prec_C X_i\mid\mathcal H_{s-1}\right\}\\
&\quad\le c_Cu(\log\log n)^d
 \left\lceil(\log n)^{d-1}e^{-\sqrt{\log\log n}/4}\right\rceil
=o\bigl(u(\log n)^{d-1}\bigr).
\end{aligned}
\]
The exceptional probability and its dependence only on observations already seen follow from Lemma~\ref{gauss:threshold-probability-estimates}.
\end{proof}

\Needspace{4\baselineskip}
\bigskip
\noindent\textbf{Cutoff estimates}

\nopagebreak

\begin{lemma}\label{gauss:cutoff-volume-estimate}
Let $C=\{x:Ax\ge0\}$ be associated with $\Sigma=AA^{\mathsf T}$ as above, in dimension $d\ge2$. For $s$ independent standard Gaussian observations, let $\mathcal U_s=\{x:\nexists i\le s\text{ with }x\prec_C X_i\}$ and $A_s=\gamma(\mathcal U_s)$. As $s\to\infty$,
\begin{equation}\label{gauss:cutoff-moments}
\E\left[A_s\right]\sim\frac{a_\Sigma(\log s)^{d-1}}s,
\qquad
\E\left[\int_{\mathcal U_s}q_C(x)\,d\gamma(x)\right]\sim\frac{a_\Sigma(\log s)^{d-1}}{s^2}.
\end{equation}
If $ma_\Sigma(\log s)^{d-1}/s=(d-1)\log\log s+O(1)$, then
\begin{equation}\label{gauss:cutoff-no-selection}
\E\left[(1-A_s)^m\right]=O_C((\log s)^{1-d}).
\end{equation}
For every fixed $\varepsilon\in(0,1)$,
\begin{equation}\label{gauss:cutoff-concentration}
\Prb\left\{A_s<(1-\varepsilon)\E\left[A_s\right]\right\}
=O_{C,\varepsilon}((\log s)^{1-d}),
\qquad \frac{A_s}{\E\left[A_s\right]}\longrightarrow1\text{ in }L^1.
\end{equation}
Moreover,
\begin{equation}\label{gauss:cutoff-waiting-estimate}
\E\left[\frac1{A_s^2}\int_{\mathcal U_s}q_C(x)\,d\gamma(x)\right]
=O_C((\log s)^{1-d}),
\end{equation}
and, for each fixed $T>0$,
\begin{equation*}
s\E\left[\frac1{A_s}\int_{\mathcal U_s}
       q_C(x)\ind\{q_C(x)\le T/s\}\,d\gamma(x)\right]
\longrightarrow\int_0^T t e^{-t}\,dt.
\end{equation*}
Changing coordinates gives the same estimates for a fixed positive definite Gaussian correlation matrix, with $F$ and $q_F$ in place of $\gamma$ and $q_C$, and with $\mathcal U_s$ expressed in the original coordinates. The coefficient remains $a_\Sigma$.
\end{lemma}

\begin{proof}
We estimate the mean and variance of the probability left in directions $P_s$. This probability is at most $A_s$ and has the same leading mean, so its concentration controls small values of $A_s$ and the random denominators.

\proofstep{Means.}
Conditioning on an independent $X$ gives
\[
\E\left[\int_{\mathcal U_s}q_C(x)^k\,d\gamma(x)\right]
=\int_0^1v^k(1-v)^s\,dH(v),\qquad k=0,1.
\]
Apply \eqref{eq:distribution-cutoff-moments} with $T=\infty$ and
$H(1/s)\sim a_\Sigma(\log s)^{d-1}/s$ to obtain
\eqref{gauss:cutoff-moments}. For fixed $T>0$, the same formula gives
\begin{equation}\label{gauss:cutoff-truncated-mean}
\E\left[\int_{\mathcal U_s}q_C(x)\ind\{q_C(x)\le T/s\}\,d\gamma(x)\right]
\sim\frac{a_\Sigma(\log s)^{d-1}}{s^2}\int_0^T t e^{-t}\,dt.
\end{equation}

\proofstep{Uniform mean bound.}
Put $\delta=1/\log\log s$, $u=(\log s)^2/s$, and
\[
V=\int_{\mathcal U_s}\ind\{x/\|x\|\in P_s,\ q_C(x)\le u\}\,d\gamma(x).
\]
We use $V\le A_s$. First,
\begin{equation}\label{gauss:cutoff-direction-coefficient}
0\le a_\Sigma-a_{P_s}=O_C(\delta).
\end{equation}
To see this, choose a unit vector $e\in\operatorname{int}(C^*)$ and $b>0$ such that $e\cdot z\ge b\|z\|$ on $C$. For $h=2\delta/b$, the map $\theta\mapsto(\theta+he)/\|\theta+he\|$ sends $C^*\cap S^{d-1}$ into $P_s$ for all sufficiently large $s$. It is injective for $h<1$, and its surface Jacobian is $(1+h\theta\cdot e)/\|\theta+he\|^d=1+O_C(\delta)$. Thus the area omitted by $P_s$ is $O_C(\delta)$. The integrand $1/J(\theta)$ is bounded, since $J(\theta)\ge e^{-1}\vol(C\cap B(0,1))$ for unit $\theta\in C^*$. Integrating over the omitted directions proves \eqref{gauss:cutoff-direction-coefficient}.

We next quantify the radial calculation used in \eqref{gauss:cone-angular-limit}. For $\theta\in P_s$,
\[
0<e^{-1}\vol(C\cap B(0,1))\le J(\theta)\le O_d(\delta^{-d}).
\]
The upper bound follows by integrating $e^{-\delta\|z\|}$ over $\mathbb R^d$. Let $t_v(\theta)$ solve $q_C(t_v\theta)=v$. Along these directions, $q_C(t\theta)$ is strictly decreasing. For every $s^{-1}(\log s)^{-22d}\le v\le s^{-1}(\log s)^2$, \eqref{gauss:cone-quantitative-laplace} gives
\[
t_v(\theta)^2=2\log s+O_C(\log\log s).
\]
The Gaussian radial tail and the same relative error bound then give
\begin{equation}\label{gauss:cutoff-quantitative-radial}
\begin{aligned}
 \phi(0)\int_{t_v(\theta)}^\infty e^{-t^2/2}t^{d-1}\,dt
 &=\frac{v\,t_v(\theta)^{2d-2}}{J(\theta)}
     \left[1+O_C\left(\frac{(\log\log s)^2}{\log s}\right)\right]\\
 &=\frac{2^{d-1}v(\log s)^{d-1}}{J(\theta)}
     \left[1+O_C\left(\frac{(\log\log s)^2}{\log s}\right)\right].
\end{aligned}
\end{equation}
Integrating over $P_s$ and then by parts gives
\begin{equation}\label{gauss:cutoff-restricted-mean}
\begin{aligned}
\E\left[V\right]
&=(1-u)^s\Prb\left\{q_C(X)\le u,\ X/\|X\|\in P_s\right\}\\
&\quad+s\int_0^u(1-v)^{s-1}
           \Prb\left\{q_C(X)\le v,\ X/\|X\|\in P_s\right\}\,dv\\
&=\frac{a_{P_s}(\log s)^{d-1}}s
 \left[1+O_C\left(\frac{(\log\log s)^2}{\log s}\right)\right]\\
&=\frac{a_\Sigma(\log s)^{d-1}}s
 \left[1+O_C\left(\frac1{\log\log s}\right)\right].
\end{aligned}
\end{equation}
The second equality uses \eqref{gauss:cutoff-quantitative-radial},
whose error holds throughout $P_s$ and the stated interval of $v$.
The endpoint is at most $e^{-(\log s)^2}$, and \eqref{gauss:F-order}
bounds the contribution below $s^{-1}(\log s)^{-22d}$ within the
error. The last equality uses \eqref{gauss:cutoff-direction-coefficient}.

\proofstep{Variance.}
Let $V^{(-i)}$ be the same integral after deleting observation $X_i$, with $P_s$ and $u$ unchanged. The difference counts points dominated by $X_i$ but by no other sample point. To bound $\E\left[(V^{(-i)}-V)^2\right]$, condition on $X_i=x_i$ and write
\[
g_{x_i}(v)=\Prb\left\{X/\|X\|\in P_s,\ q_C(X)\le v,\ x_i-X\in C\right\}.
\]
Two points $x,x'$ both contribute to $V^{(-i)}-V$ only if the other $s-1$ observations avoid their combined dominating region. That region has probability at least $\max\{q_C(x),q_C(x')\}$. Integrating $e^{-(s-1)\max\{q_C(x),q_C(x')\}}$ over these pairs therefore gives
\begin{equation}\label{gauss:cutoff-deletion-bound}
\begin{aligned}
\E\left[(V^{(-i)}-V)^2\mid X_i=x_i\right]
&\le e^{-(s-1)u}g_{x_i}(u)^2\\
&\quad+(s-1)\int_0^u e^{-(s-1)v}g_{x_i}(v)^2\,dv.
\end{aligned}
\end{equation}
Here we integrate by parts against the distribution $g_{x_i}(v)^2$ of the larger of the two $q_C$ values.

Suppose $\|x_i\|\le\sqrt{2\log s+40d\log\log s}$. Lemma~\ref{gauss:covered-probability} gives, for $s^{-1}(\log s)^{-22d}\le v\le u$,
\[
g_{x_i}(v)\le c_Cv\,\ind\{q_C(x_i)\le v\}
                   (1+\log(v/q_C(x_i)))^d.
\]
For smaller $v$ the left side is zero. Indeed, the norm bound and \eqref{gauss:cone-uniform-lower} give $q_C(x_i)\ge c_Cs^{-1}(\log s)^{-20d-d/2}>s^{-1}(\log s)^{-22d}$ for all large $s$, and $x_i-X\in C$ requires $q_C(x_i)\le q_C(X)$. For each fixed integer $k\ge1$, integration of the logarithmic tail gives
\begin{equation*}
\begin{aligned}
&\E\left[\ind\{q_C(X)\le v\}(1+\log(v/q_C(X)))^k\right]\\
&\quad=H(v)+k\int_0^\infty(1+t)^{k-1}H(ve^{-t})\,dt\\
&\quad\le c_{C,k}v(1+\log(1/v))^{d-1}.
\end{aligned}
\end{equation*}
By \eqref{gauss:F-order} the integral is finite. Taking $k=2d$ gives
\[
\E\left[g_{X_i}(v)^2
 \ind\{\|X_i\|\le\sqrt{2\log s+40d\log\log s}\}\right]
\le c_Cv^3(\log s)^{d-1},\qquad 0<v\le u.
\]
Substitution into \eqref{gauss:cutoff-deletion-bound} bounds each term by $O_C((\log s)^{d-1}/s^3)$.

Outside the norm bound, use $V^{(-i)}-V\le H(u)=O_C((\log s)^{d+1}/s)$. The radial-tail calculation in the proof of Lemma~\ref{gauss:threshold-probability-estimates}, with $n=s$, bounds the probability of this event by $O_d(s^{-1}(\log s)^{(d-2)/2-20d})$. Its total contribution is $o((\log s)^{d-1}/s^2)$. Efron--Stein~\cite{efron1981} now gives
\begin{equation}\label{gauss:cutoff-patch-variance}
\operatorname{Var}(V)
\le\sum_{i=1}^s\E\left[(V^{(-i)}-V)^2\right]
\le s\,O_C\left(\frac{(\log s)^{d-1}}{s^3}\right)
 +o\left(\frac{(\log s)^{d-1}}{s^2}\right)
=O_C\left(\frac{(\log s)^{d-1}}{s^2}\right).
\end{equation}

\proofstep{Exponential moment and concentration.}
For any nonnegative random variable $Y$ with $\E\left[Y\right]=1$,
\begin{equation*}
\E\left[e^{-tY}\right]\le e^{-t}
 +\bigl(1-(1+t)e^{-t}\bigr)\operatorname{Var}(Y)
\le e^{-t}+\operatorname{Var}(Y),\qquad t\ge0.
\end{equation*}
Indeed, the quadratic polynomial agreeing with $e^{-tx}$ at $x=0,1$
and with its derivative at $x=1$ is an upper bound for $x\ge0$,
because the third derivative of $e^{-tx}$ is nonpositive.
Taking expectations gives the first inequality.
Apply it to $Y=V/\E\left[V\right]$ and $t=m\E\left[V\right]$. Since $V\le A_s$,
\eqref{gauss:cutoff-restricted-mean} and \eqref{gauss:cutoff-patch-variance}
give
\[
\E\left[(1-A_s)^m\right]\le\E\left[e^{-mV}\right]
\le e^{-m\E\left[V\right]}+\frac{\operatorname{Var}(V)}{(\E\left[V\right])^2}
=O_C((\log s)^{1-d})
\]
when $ma_\Sigma(\log s)^{d-1}/s=(d-1)\log\log s+O(1)$.
This proves \eqref{gauss:cutoff-no-selection}.

For fixed $\varepsilon>0$, the same mean and variance bounds give
\[
\Prb\left\{A_s<(1-\varepsilon)\E\left[A_s\right]\right\}
\le\Prb\left\{V<(1-\varepsilon)\E\left[A_s\right]\right\}
=O_{C,\varepsilon}((\log s)^{1-d})
\]
by Chebyshev's inequality. Since $Y=A_s/\E\left[A_s\right]$ has mean one,
\[
\E\left[|Y-1|\right]=2\E\left[(1-Y)_+\right]
\le2\varepsilon+2\Prb\left\{Y<1-\varepsilon\right\}
=2\varepsilon+o(1).
\]
Let $s\to\infty$ and then $\varepsilon\downarrow0$.

Apply \eqref{eq:undominated-denominator-bound} with $\gamma$ and $q_C$. Equation~\eqref{gauss:cutoff-moments} bounds its first term, and \eqref{gauss:cutoff-concentration} with $\varepsilon=1/2$ bounds its second. Both are $O_C((\log s)^{1-d})$, proving \eqref{gauss:cutoff-waiting-estimate}.

Finally,
\[
\begin{aligned}
s\E\left[\frac1{A_s}\int_{\mathcal U_s}q_C\ind\{q_C\le T/s\}\,d\gamma\right]
&=\frac{s}{\E\left[A_s\right]}\E\left[\int_{\mathcal U_s}q_C\ind\{q_C\le T/s\}\,d\gamma\right]+o(1)\\
&\longrightarrow\int_0^T te^{-t}\,dt.
\end{aligned}
\]
The replacement error is at most $T\E\left[|A_s/\E\left[A_s\right]-1|\right]=o(1)$,
since $s\int_{\mathcal U_s}q_C\ind\{q_C\le T/s\}\,d\gamma\le T A_s$.
The limit follows from \eqref{gauss:cutoff-truncated-mean} and
\eqref{gauss:cutoff-moments}.
\end{proof}

\Needspace{4\baselineskip}
\bigskip
\noindent\textbf{Random denominators}

\nopagebreak

\begin{lemma}\label{lem:cutoff-common-replacements}
For fixed positive definite $\Sigma$, as $s\to\infty$,
\[
\E\left[(1-A_s)^m\right]=e^{-m\E\left[A_s\right]}+o_\Sigma(1),\qquad
\E\left[\frac{A_{s+\ell}}{A_s}\right]=\frac{\E\left[A_{s+\ell}\right]}{\E\left[A_s\right]}+o_\Sigma(1)
\]
for all integers $m\ge1$ and $\ell\ge0$, with errors independent of $m,\ell$. Moreover,
\begin{align}
\E\left[\sum_{j=1}^m(1-A_s)^{j-1}(A_{s+m-j}-A_{s+m})\right]
 &=O_\Sigma((\log s)^{1-d}),\label{eq:cutoff-volume-replacement}\\
\E\left[\left(1-\frac{A_{s+m}}{A_s}\right)(1-A_s)^m\right]
 &=O_\Sigma((\log s)^{1-d}).\label{eq:cutoff-volume-cross-term}
\end{align}
The same bounds hold for the cone model in
Lemma~\ref{gauss:cutoff-volume-estimate}, with constants depending on $C$.
\end{lemma}
\begin{proof}
Lemma~\ref{gauss:cutoff-volume-estimate} gives
$\E\left[|A_s/\E\left[A_s\right]-1|\right]=o_\Sigma(1)$. The elementary bounds
\[
0\le e^{-mx}-(1-x)^m\le c x,\qquad
|e^{-ty}-e^{-t}|\le c|y-1|
\quad(0\le x\le1,\ y,t\ge0)
\]
follow from $-\log(1-x)=x+O(x^2)$ for $x\le1/2$ and the mean value theorem for $y\ge1/2$; the remaining cases are bounded directly by one.
They imply, for all $m\ge1$,
\[
\left|\E\left[(1-A_s)^m\right]-e^{-m\E\left[A_s\right]}\right|
\le c\E\left[A_s\right]+c\E\left[|A_s/\E\left[A_s\right]-1|\right]=o_\Sigma(1).
\]
Similarly,
\[
\left|\E\left[\frac{A_{s+\ell}}{A_s}\right]-\frac{\E\left[A_{s+\ell}\right]}{\E\left[A_s\right]}\right|
\le\E\left[\frac{A_{s+\ell}}{A_s}
 \left|1-\frac{A_s}{\E\left[A_s\right]}\right|\right]
\le\E\left[\left|1-\frac{A_s}{\E\left[A_s\right]}\right|\right].
\]
For the last two assertions, conditional on $Z_{1:s}$, each additional
observation removes expected probability at most $\int_{\mathcal U_s}q\,dF$.
The difference $A_{s+m-j}-A_{s+m}$ involves $j$ such observations,
whereas $A_s-A_{s+m}$ involves $m$. Therefore the left sides of
\eqref{eq:cutoff-volume-replacement} and
\eqref{eq:cutoff-volume-cross-term} are bounded, respectively, by
\[
\begin{aligned}
\E\left[\int_{\mathcal U_s}q\,dF\sum_{j=1}^m j(1-A_s)^{j-1}\right]
&\le\E\left[A_s^{-2}\int_{\mathcal U_s}q\,dF\right],\\
\E\left[\int_{\mathcal U_s}q\,dF\frac{m(1-A_s)^m}{A_s}\right]
&\le\frac1e\E\left[A_s^{-2}\int_{\mathcal U_s}q\,dF\right].
\end{aligned}
\]
The inequalities use $\sum_{j\ge1}j(1-A_s)^{j-1}=A_s^{-2}$ and
$mA_s(1-A_s)^m\le1/e$. Apply \eqref{gauss:cutoff-waiting-estimate} to finish.
\end{proof}

\Needspace{4\baselineskip}
\bigskip
\noindent\textbf{Marginal tail estimates}

\nopagebreak
From this point on, $d$ denotes the number of observed coordinates.
For a singular Gaussian covariance matrix, we write $r$ for its rank.

For a binomial count $N$ and an integer $k\ge\E\left[N\right]$,
\begin{equation}\label{eq:binomial-count-tail}
\Prb\left\{N\ge k\right\}\le(e\E\left[N\right]/k)^k.
\end{equation}

\begin{lemma}\label{lem:marginal-count-estimates}
For any continuous marginals, outside an event of probability $O_d(n^{-1})$ determined by the initial $N=\lfloor n/2\rfloor$ observations, their upper-tail probabilities $p$ and counts $j$ satisfy, simultaneously at every coordinate value,
\begin{equation}\label{gauss:binomial-table-accuracy}
0\le\overline p_N(j)-p\le
c\left(\sqrt{\frac{p(1-p)\log n}{N}}+\frac{\log n}{N}\right).
\end{equation}
No independence between coordinates is required.
\end{lemma}
\begin{proof}
For one coordinate, let $N=\lfloor n/2\rfloor$ and transform the initial observations to their uniform upper-tail probabilities, with order statistics $V_{(1)},\ldots,V_{(N)}$. If a candidate has true upper-tail probability $p$ and count $j<N$, then $p\le V_{(j+1)}$. By the binomial identity,
\[
\Prb\left\{V_{(j+1)}>\overline p_N(j)\right\}
=\Prb\left\{\operatorname{Bin}(N,\overline p_N(j))\le j\right\}=n^{-2}.
\]
For $j=N$ the bound holds because $\overline p_N(N)=1$. A union bound over all ranks and coordinates gives exceptional probability at most $dN/n^2$. Outside this event, $p\le\overline p_N(j)$ for every candidate position at once.

For accuracy, let $F_N^-(p)=N^{-1}\#\{i:V_i<p\}$.
Bernstein's inequality on the grid $p=i/n^3$, a union bound, and
monotonicity between adjacent grid points give, with probability
$1-O(n^{-3})$,
\begin{equation*}
|F_N^-(p)-p|\le c\left(\sqrt{\frac{p(1-p)\log n}{N}}
+\frac{\log n}{N}\right),\qquad 0\le p\le1.
\end{equation*}
The grid error is absorbed by $(\log n)/N$.

Inverting the lower-tail Bernstein bound also gives, for $r=j/N$ and an absolute constant $c$,
\begin{equation*}
0\le \overline p_N(j)-r
\le c\left(\sqrt{\frac{r(1-r)\log n}{N}}
                   +\frac{\log n}{N}\right).
\end{equation*}
Combine these inequalities with $r=F_N^-(p)$ and
$|r(1-r)-p(1-p)|\le|r-p|$ to obtain
\eqref{gauss:binomial-table-accuracy} simultaneously in all coordinates.

For fixed $\kappa>0$, when $p,1-p\ge n^{-1+\kappa}$, the error relative to either tail is $O(n^{-\kappa/2}\sqrt{\log n}+n^{-\kappa}\log n)$. Writing $\phi$ for the standard normal density, $1/\phi(\Phi^{-1}(u))\le O(1)/\min(u,1-u)$. Thus the estimated and true standardized coordinates differ by $o(1/\sqrt{\log n})$ for every point in the stated range.

\end{proof}

\Needspace{4\baselineskip}
\bigskip
\noindent\textbf{Correlation estimates}

\nopagebreak

\begin{lemma}\label{gauss:conservative-estimates}
Fix $d\ge2$ and a positive definite Gaussian correlation matrix $\Sigma$. For the estimates in \eqref{gauss:conservative-correlation} and \eqref{gauss:conservative-rank-score}, there is an event determined by the first $\lfloor n/2\rfloor$ observations, of probability $1-O_d(n^{-1})$, on which, for all sufficiently large $n$,
\begin{gather*}
\overline{\Sigma}\text{ is positive definite},\qquad
\overline{\Sigma}_{k\ell}\ge \Sigma_{k\ell},\qquad
\|\overline{\Sigma}-\Sigma\|=O_d(\sqrt{\log n/n}),\\
a_{\overline{\Sigma}}/a_\Sigma=1+o(1),\qquad
q(z)\le\widehat q(z)\quad\text{for every }z.
\end{gather*}
For each fixed $\kappa>0$, the same event gives $\widehat q(z)/q(z)=1+o(1)$ for all points whose standardized normal coordinates have norm $O(\sqrt{\log n})$ and whose two marginal tails are at least $n^{-1+\kappa}$. The error tends to zero independently of the point. The marginal estimates in the proof do not require independence of coordinates or a positive definite covariance.
\end{lemma}

\begin{proof}
Use the marginal bounds and the standardized-coordinate error proved
in Lemma~\ref{lem:marginal-count-estimates}.

\proofstep{Correlation estimates.}
For this estimate, let $N=\lfloor n/4\rfloor$. The difference of two observations has the same correlation matrix. In standard Gaussian coordinates the two projection vectors make angle $\arccos(\Sigma_{k\ell})$. Their planar direction is uniform, and opposite signs occupy fraction $\arccos(\Sigma_{k\ell})/\pi$ of the circle. The sign-agreement probability is therefore $p_{k\ell}=1-\arccos(\Sigma_{k\ell})/\pi$. Disjoint pairs are independent, so $N_{k\ell}\sim\operatorname{Bin}(N,p_{k\ell})$. Inverting the one-sided binomial bound gives $\Prb\left\{\overline p_N(N_{k\ell})<p_{k\ell}\right\}\le n^{-2}$. Moreover, Hoeffding's inequality gives
\[
\overline p_N(j)\le\min\{1,j/N+\sqrt{\log n/N}\},\qquad
|N_{k\ell}/N-p_{k\ell}|=O(\sqrt{\log n/N})
\]
simultaneously for all coordinate pairs except with probability $O_d(n^{-2})$ for the second bound. Since $-\cos(\pi p)$ is increasing and Lipschitz, the correlation estimates have the stated entrywise ordering and matrix error. Since $\Sigma$ is fixed and positive definite, $\overline{\Sigma}$ is also positive definite for large $n$. Continuity of the positive integral \eqref{gauss:sigma-orthant-coefficient} gives $a_{\overline{\Sigma}}/a_\Sigma\to1$.

\proofstep{Gaussian rectangle probabilities.}
Work in standardized marginal coordinates in this step, and write
\[
q_S(z)=\Prb_{W\sim \mathcal N_d(0,S)}
 \left\{W^{(k)}\ge z^{(k)}\text{ for every }k\in[d]\right\}
\]
when the correlation matrix $S$ must be specified. If $S\ge \Sigma$ entrywise off the diagonal, then $q_S(z)\ge q_\Sigma(z)$. For completeness, interpolate the covariance and approximate the rectangle indicator by $g(x)=\prod_k g_k(x^{(k)})$, with smooth nonnegative increasing functions $g_k$. The Gaussian differentiation formula (Price's theorem; see \cite{voigtlaender2021}) gives
\[
\frac{d}{dt}\E\left[g(W_t)\right]
=\frac12\sum_{k,\ell}(S-\Sigma)_{k\ell}
                    \E\left[\partial_{k\ell}g(W_t)\right],\qquad
W_t\sim \mathcal N_d(0,(1-t)\Sigma+tS).
\]
Diagonal terms vanish and all remaining terms are nonnegative. Taking limits proves the bound, also when some thresholds are $-\infty$. Together with the marginal bounds this gives
\[
q_\Sigma(z)\le q_\Sigma(\widehat z)\le
q_{\overline{\Sigma}}(\widehat z)=\widehat q(z).
\]

Finally, fix constants $B,\kappa>0$ and consider points satisfying
$\|z\|\le B\sqrt{\log n}$ and
$\min\{\Phi(z^{(k)}),1-\Phi(z^{(k)})\}\ge n^{-1+\kappa}$
for every $k\in[d]$.
Lemma~\ref{lem:marginal-count-estimates} gives
$\|\widehat z-z\|=o(1/\sqrt{\log n})$ for all these points, and
$\|\overline{\Sigma}-\Sigma\|\log n=o(1)$.
If $f_S$ denotes the density of $\mathcal N_d(0,S)$, then
\[
q_S(z)=\int_{[0,\infty)^d}f_S(z+v)\,dv.
\]
We compare the two integrals first on a ball and then control their tails. Since the eigenvalues of $\overline\Sigma$ stay in a fixed compact subset of $(0,\infty)$, expansion of the Gaussian log densities gives, for every fixed $R$,
\begin{equation}\label{gauss:local-density-ratio}
\sup_{\substack{
 \|z\|\le B\sqrt{\log n},\ v\ge0\\
 \|z+v\|\le R\sqrt{\log n}\\
 \min_{k\in[d]}\min\{\Phi(z^{(k)}),1-\Phi(z^{(k)})\}
 \ge n^{-1+\kappa}
}}
\left|
 \log\frac{f_{\overline\Sigma}(\widehat z+v)}
               {f_\Sigma(z+v)}
\right|=o(1).
\end{equation}
The log-density error is $O_\Sigma(\|\overline\Sigma-\Sigma\|\log n+\|\widehat z-z\|\sqrt{\log n})=o(1)$. This comparison uses only the covariance and translation errors. It holds on the same ball for any measurable integration region, not just $v\ge0$.

Next integrate $f_\Sigma(z+v)$ over the fixed cube $v\in[0,1]^d$. There are constants $K=K(B,\Sigma)$ and $c=c(B,\Sigma)>0$ such that
\begin{equation}\label{gauss:rectangle-polynomial-lower}
q_\Sigma(z)\ge c n^{-K}
\qquad\text{whenever }\|z\|\le B\sqrt{\log n}.
\end{equation}
Choose $R$ so large that the Gaussian tails outside this ball are
$O(n^{-K-2})$, uniformly for covariance matrices near $\Sigma$.
By \eqref{gauss:rectangle-polynomial-lower}, both discarded integrals
are $o(q_\Sigma(z))$. Integrating \eqref{gauss:local-density-ratio}
on the ball therefore proves
$q_{\overline\Sigma}(\widehat z)/q_\Sigma(z)=1+o(1)$ uniformly over the stated points.
\end{proof}

To apply the relative estimate in the stopping proof, choose $P$ inside $C^*$ and avoid the finitely many directions $\pm a_k$, where $a_k$ are the unit rows of $\Sigma^{1/2}$. Their surface measure is zero when $d\ge2$, so compact sets avoiding small neighborhoods of these directions can have $a_P$ arbitrarily close to $a_\Sigma$. On such a fixed set, $|a_k\cdot\theta|\le b<1$ for all $k$. For $\|X\|\le\sqrt{2\log n+40d\log\log n}$, both marginal tails are therefore at least $n^{-1+\kappa}$ for some fixed $\kappa>0$. Lemma~\ref{gauss:conservative-estimates} gives relative accuracy there. The probability beyond this radius is $O_d(n^{-1}(\log n)^{(d-2)/2-20d}) =o(u(\log n)^{d-1})$ throughout \eqref{gauss:smooth-threshold-band}.

\Needspace{4\baselineskip}
\bigskip
\noindent\textbf{Known correlation matrix}

\nopagebreak

\begin{lemma}\label{gauss:support-estimates}
Suppose the correlation matrix $\Sigma=AA^{\mathsf T}$ is known,
where $A\in\mathbb R^{d\times r}$ has full column rank $r\ge2$
and unit rows $a_1^{\mathsf T},\ldots,a_d^{\mathsf T}$.
Assume that $C=\{x\in\mathbb R^r:Ax\ge0\}$ has nonempty interior. Use the marginal estimates in \eqref{gauss:conservative-rank-score} with $\Sigma$ in place of $\overline\Sigma$. Outside an event of probability $O_d(n^{-1})$ determined by the initial sample, $q(z)\le\widehat q(z)$ for every $z$. For any fixed compact set of directions inside $C^*$ avoiding the unit row directions $\pm a_k$, the same event gives $\widehat q(Ax)=(1+o(1))q_C(x)$ when $\|x\|\le\sqrt{2\log n+40r\log\log n}$. The error is uniform over these points.
\end{lemma}
\begin{proof}
Lemma~\ref{lem:marginal-count-estimates} gives the marginal confidence bounds without any assumption on dependence. Applying them with the known $\Sigma$ gives the first assertion. Avoiding the row directions gives the same lower bounds on both marginal tails as above. The standardized coordinate errors are therefore at most $\delta=o(1/\sqrt{\log n})$. Choose a fixed $v$ with $(Av)^{(k)}\ge1$ for every coordinate $k$. Such a vector exists by scaling an interior point of $C$, since no row of $A$ is zero. Then
\[
x+C\ \subseteq\ \{y:Ay\ge\widehat z\}
\ \subseteq\ x-\delta v+C,\qquad z=Ax.
\]
Hence $q_C(x)\le\widehat q(z)\le q_C(x-\delta v)$.
The Gaussian log-density expansion used to prove
\eqref{gauss:local-density-ratio}, now in dimension $r$ with covariance
$I_r$ and translation $-\delta v$, gives a log-density error
$O_{v,R}(\delta\sqrt{\log n}+\delta^2)=o(1)$ on
$\|y\|\le R\sqrt{\log n}$ for every fixed $R$.
Integrating over $x+C$ inside this ball therefore changes its
probability by a factor $1+o(1)$.
Equation~\eqref{gauss:cone-uniform-lower} gives a polynomial lower bound
on $q_C(x)$ throughout the stated range. Choosing $R$ sufficiently
large makes both discarded Gaussian tails negligible relative to
$q_C(x)$. Thus $q_C(x-\delta v)/q_C(x)=1+o(1)$, proving the claim.
\end{proof}

\Needspace{4\baselineskip}
\bigskip
\noindent\textbf{Conditional selection bounds}

\nopagebreak

The preceding estimates give the same two bounds whether a rule uses $q$ itself or estimates it from the initial sample.

\begin{corollary}\label{lem:threshold-common-inputs}
Let $Z=AX$, where $X\sim\mathcal N_r(0,I_r)$, $r\ge2$, $A$ has column rank $r$ and unit rows, and $C=\{x:Ax\ge0\}$ has nonempty interior. Write $\Sigma=AA^{\mathsf T}$ and let $a_\Sigma$ be the coefficient in \eqref{gauss:cone-tail-exact}, in dimension $r$. For full information, set $\widehat q=q$. Otherwise use the estimate from Lemma~\ref{gauss:support-estimates} when $\Sigma$ is known, or Lemma~\ref{gauss:conservative-estimates} when $\Sigma$ is unknown and positive definite. Fix these estimates after the first $\lfloor n/2\rfloor$ observations.

For every fixed $\varepsilon\in(0,1)$ and all sufficiently large $n$, outside exceptional histories of total probability $o(((\log n)^{r-1}\log\log n)^{-1})$, we have $q\le\widehat q$ and
\begin{equation*}
\begin{aligned}
\Prb\left\{Z_s\in\Max(Z_1,\ldots,Z_s),\ \widehat q(Z_s)\le u
       \mid\mathcal H_{s-1}\right\}
&\ge(1-\varepsilon)a_\Sigma(\log n)^{r-1}u,\\
\E\left[q(Z_s)\ind\{\widehat q(Z_s)\le u\}\mid\mathcal H_{s-1}\right]
&\le(1+\varepsilon)\frac{a_\Sigma}{2}(\log n)^{r-1}u^2.
\end{aligned}
\end{equation*}
Here $\mathcal H_{s-1}$ is the information defined in Section~\ref{sec:model}, $s>\lfloor n/2\rfloor$, and
\[
\frac1{n(\log n)^{2(r-1)}}\le u\le\frac{e^{-\sqrt{\log\log n}/3}}n.
\]
The bounds hold for all these times, thresholds, and histories. Each exception is determined by observations already seen. The same exceptional histories work for every threshold in this interval, including a threshold chosen using the initial sample. For independent coordinates, take $A=I_d$, $r=d$, and $a_I=1/(d-1)!$. Then the known-matrix estimate is exactly the product in \eqref{eq:threshold-rank-score}. For $\widehat q=q$, the bounds also hold for standard Gaussian observations compared by any fixed closed convex cone with nonempty interior and no line. Use $q_C$ and replace the maximality condition by $\nexists i<s:X_s\prec_C X_i$, and replace $a_\Sigma$ by the coefficient in \eqref{gauss:cone-tail-exact}.
\end{corollary}
\begin{proof}
Choose a fixed compact set $P\subset\operatorname{int}(C^*)\cap S^{r-1}$, avoiding the finitely many unit row directions of $A$, with $a_P>(1-\varepsilon/8)a_\Sigma$. The two estimation lemmas, together with the direction argument after Lemma~\ref{gauss:conservative-estimates}, give $q\le\widehat q$ everywhere and $\widehat q(AX)\le(1+\varepsilon/8)q_C(X)$ on $P$ within the norm bound $\|X\|\le\sqrt{2\log n+40r\log\log n}$, for large $n$. For $\widehat q=q$ these statements are immediate, and for a general cone no row directions need to be excluded. The probability beyond this norm bound is $o(u(\log n)^{r-1})$ throughout the stated interval. Since the new observation is independent of $\mathcal H_{s-1}$,
\[
\begin{aligned}
&\Prb\left\{Z_s\in\Max(Z_1,\ldots,Z_s),\ \widehat q(Z_s)\le u
          \mid\mathcal H_{s-1}\right\}\\
&\quad\ge\Prb\left\{\widehat q(Z_s)\le u,\ X_s/\|X_s\|\in P
          \mid\mathcal H_{s-1}\right\}\\
&\qquad-\Prb\left\{q_C(X_s)\le u,\ X_s/\|X_s\|\in P,\
 \exists i<s:\ X_s\prec_C X_i\mid\mathcal H_{s-1}\right\}\\
&\quad\ge(1-\varepsilon/2)a_\Sigma u(\log n)^{r-1}
       -o(u(\log n)^{r-1})\\
&\quad\ge(1-\varepsilon)a_\Sigma u(\log n)^{r-1}.
\end{aligned}
\]
The first inequality uses $q\le\widehat q$. The second uses
\eqref{gauss:patch-sharp-moments} at $u/(1+\varepsilon/8)$ and
Lemma~\ref{gauss:sample-exclusion} for points dominated by earlier
observations. For the other bound,

\[
\E\left[q(Z_s)\ind\{\widehat q(Z_s)\le u\}\mid\mathcal H_{s-1}\right]
\le\E\left[q(Z_s)\ind\{q(Z_s)\le u\}\right]
\le(1+\varepsilon)\frac{a_\Sigma}{2}u^2(\log n)^{r-1},
\]
where the last inequality is the truncated-moment estimate in \eqref{gauss:patch-sharp-moments}. The exceptional histories come from Lemma~\ref{gauss:sample-exclusion} and, when estimating $q$, the initial-sample event in the relevant estimation lemma. Their combined probability has the stated order. For $\widehat q=q$, no estimation lemma is needed, and the same proof works for every cone stated above.
\end{proof}

\end{document}